\documentclass[11pt,letterpaper]{article}

\usepackage[english]{babel}
\usepackage[utf8x]{inputenc}
\usepackage{amsmath}
\usepackage{amsthm, amssymb}
\usepackage{amsfonts}
\usepackage{algorithm}
\usepackage{algorithmic}
\usepackage{graphicx}
\usepackage[dvipsnames]{color}
\usepackage{fullpage}
\usepackage{bbm}
\usepackage[
  pagebackref,
  colorlinks=true,
  urlcolor=blue,
  linkcolor=blue,
  citecolor=blue,
]{hyperref}
\usepackage[nameinlink]{cleveref}

\usepackage{color-edits}
\addauthor{hb}{magenta}

\newtheorem{theorem}{Theorem}

\newtheorem{definition}{Definition}
\newtheorem{lemma}{Lemma}
\newtheorem{corollary}{Corollary}
\newtheorem{proposition}{Proposition}

\newtheorem{observation}{Observation}
\newtheorem{fact}{Fact}

\newcommand{\sre}{{X_S(n,c)}}
\newcommand{\bbm}{{X_B(n,\ell)}}
\newcommand{\lbm}{{X_L(n,m)}}
\newcommand{\erm}{\mathcal{ER}^m}
\newcommand{\er}{\mathcal{ER}}

\newcommand{\rev}{\textsc{Rev}}
\newcommand{\srev}{\textsc{SRev}}
\newcommand{\vcg}{\textsc{VCG}}

\begin{document}
\title{Optimal (and Benchmark-Optimal) Competition Complexity for Additive Buyers over Independent Items}
\author{Hedyeh Beyhaghi
\thanks{Cornell University, \tt{hedyeh@cs.cornell.edu}} \and S. Matthew Weinberg%
\thanks{Princeton University, \tt{smweinberg@princeton.edu}. Supported by NSF CCF-1717899.}} 

\maketitle

\begin{abstract}
The \emph{Competition Complexity} of an auction setting refers to the number of additional bidders necessary in order for the (deterministic, prior-independent, dominant strategy truthful) Vickrey-Clarke-Groves mechanism to achieve greater revenue than the (randomized, prior-dependent, Bayesian-truthful) optimal mechanism without the additional bidders. 

We prove that the competition complexity of $n$ bidders with additive valuations over $m$ independent items is at most $n(\ln(1+m/n)+2)$, and also at most $9\sqrt{nm}$. When $n \leq m$, the first bound is optimal up to constant factors, even when the items are i.i.d. and regular. When $n \geq m$, the second bound is optimal for the benchmark introduced in~\cite{EdenFFTW17a} up to constant factors, even when the items are i.i.d. and regular. We further show that, while the Eden et al. benchmark is not necessarily tight in the $n \geq m$ regime, the competition complexity of $n$ bidders with additive valuations over even $2$ i.i.d. regular items is indeed $\omega(1)$. 

Our main technical contribution is a reduction from analyzing the Eden et al. benchmark to proving stochastic dominance of certain random variables. 

\end{abstract}
\addtocounter{page}{-1}
\newpage


\section{Introduction}
In the past decade, the TCS community has made radical progress developing the theory of multi-dimensional mechanism design. In particular, it was previously well-understood the optimal multi-item auctions are prohibitively complex, even with just $m=2$ items, and even subject to fairly restricted instances~\cite{BriestCKW15, HartN13, HartR15, Thanassoulis04, Pavlov11, DaskalakisDT17}. Yet, starting from seminal work of Chawla, Hartline, and Kleinberg~\cite{ChawlaHK07}, a large body of work now proves that simple auctions are in fact approximately optimal in quite general settings~\cite{ChawlaHMS10, ChawlaMS15, HartN17,LiY13, BabaioffILW14, Yao15, RubinsteinW15, ChawlaM16, CaiZ17, EdenFFTW17b}, helping to explain the prevalence of simple auctions in practice. Still, it would be a reach to claim that this agenda is convincingly resolved.

In particular, the thought of settling for $50\%$ (or even $90\%$) of the optimal achievable revenue may be a non-starter for high-stakes auctions. Indeed, there are no hard constraints forcing the auctioneer to use a simple auction. Still, \emph{Prior-independent} auctions are desirable since they don't require the auctioneer to understand the population from which consumers are drawn. \emph{Deterministic} and \emph{Dominant Strategy Truthful} auctions are desirable because consumers' strategic behavior is easier to predict. \emph{Computationally tractable} auctions are desirable because they can be efficiently found. On the other hand, it is hard to imagine that auctioneers stand a hard line on simplicity if additional market research or outsourcing computation would increase revenues, even modestly. 

The resource augmentation paradigm takes a different view: spend effort recruiting additional bidders rather than carefully designing a superior auction. We are therefore interested in answering the following question: for a given auction setting, \emph{how many additional bidders} are necessary for a simple auction (with additional bidders) to guarantee greater expected revenue  than the optimal (without)? Eden et al. term the answer to this question the \emph{competition complexity}~\cite{EdenFFTW17a}. 

This question was first studied in seminal work by Bulow and Klemperer in the context of single-item auctions~\cite{BulowK96}. Remarkably, they show that just a single additional bidder suffices for the second-price auction to guarantee greater expected revenue than Myerson's optimal auction~\cite{Myerson81} (without the additional bidder), subject to a technical condition on the population called regularity. For multi-item auctions, similar results have even more bite, as the optimal multi-item auction is considerably more complex than Myerson's (which is deterministic, dominant strategy truthful, and computationally tractable, but not prior-independent). Our main result is optimal bounds on the competition complexity for the core setting of additive bidders with independent items. Specifically,\\

\noindent\textbf{Main Result:} The competition complexity of $n$ bidders with additive values over $m$ independent items is at most $n(2+ \ln(1+m/n))$, and also $9\sqrt{nm}$. When $n \leq m$ the first bound is tight (up to constant factors). When $n \geq m$, the second bound is tight (up to constant factors) for any argument that starts from the benchmark introduced in~\cite{EdenFFTW17a}.

\subsection{Brief Technical Overview}
Formally, we consider $n$ bidders drawn independently from a distribution $D$. We study the now-standard setting of additive bidders over $m$ independent items: each bidder's value $v_j$ for item $j$ is drawn independently from some single-variate distribution $D_j$, and her value for a set $S$ of items is $\sum_{j \in S} v_j$. The simple mechanism we study is to sell the items separately, either via the second-price auction in the case of regular distributions, or Myerson's optimal single-item auction in the general case.\footnote{For irregular distributions, it is known that no guarantees are possible with prior-independence, even for a single item. The example to have in mind is a distribution with a point mass at $p$ with probability $1/p$ and $0$ otherwise: as $p \rightarrow \infty$, any auction that achieves revenue close to optimum must sell the item to the bidder with value $p$ for price close to $p$ whenever there is exactly one. It is impossible to have a single auction that does this for all $p$.} Observe that, since the bidders are additive and values are independent, selling the items separately is really just $m$ separate single-item problems. We are interested in understanding the minimum $c(n,m)$ such that selling separately to $n+c(n,m)$ bidders drawn from $D$ yields greater expected revenue than the optimal mechanism with $n$ bidders drawn from $D$ for \emph{any} $D = \times_{j=1}^m D_j$. 

Our approach starts from the benchmark proposed in~\cite{EdenFFTW17a}. That is, Eden et al. propose an upper bound on the optimal achievable revenue with $n$ bidders drawn from $D$ via the duality framework of~\cite{CaiDW16}.\footnote{We note that this upper bound can also be derived without duality using techniques of~\cite{ChawlaMS15}.} We defer a definition of this benchmark to Section~\ref{sec:duality}: it defines a \emph{Virtual Value} $\Phi_{j}(\vec{v}_i)$ of a bidder with values $\vec{v}_i$ for item $j$, and upper bounds the optimal expected revenue with $\mathbb{E}[\sum_j \max_{i \in [n]} \{\Phi_j(\vec{v}_i)\}]$. We defer most details to the technical sections, but briefly note that at this point, our analysis diverges from prior work. Eden et al. use an elegant coupling argument to connect this benchmark to the expected revenue of selling separately with additional bidders~\cite{EdenFFTW17a}. The high-level distinction in our approach is a significantly more in-depth analysis of this benchmark. In particular, our analysis makes more extensive use of Myersonian virtual value theory (Sections~\ref{sec:single} and~\ref{sec:multi}), which reduces the problem to questions purely regarding whether various methods of drawing correlated values from $[0,1]$ stochastically dominate one another (Section~\ref{sec:dominance}). 

\subsection{Connection to Related Works}
The two works most directly related to ours are~\cite{EdenFFTW17a} and~\cite{FeldmanFR18}. The one-sentence distinction between our results and these is that we strictly improve their main results regarding selling separately to additive bidders with independent items, but do not address alternative settings. For example, this paper contains no results beyond additive bidders (considered in~\cite{EdenFFTW17a}), or results for mechanisms aside from selling separately (considered in~\cite{FeldmanFR18}).   \\

\noindent\textbf{``Little $n$ Regime'':} For $n$ additive bidders with $m = \Omega(n)$ independent items, Eden et al.~\cite{EdenFFTW17a} prove a competition complexity bound of $n+2(m-1)$. Feldman et al.~\cite{FeldmanFR18} prove that selling separately to $O(n\ln(m/n)/\varepsilon)$ additional buyers exceeds a $(1-\varepsilon)$ fraction of the optimal revenue (without the additional buyers). Our main result essentially achieves the greatly improved bound of~\cite{FeldmanFR18} (and improves it further) without losing \emph{any} revenue: we prove a competition complexity bound of $n(2+\ln(1+m/n))$. This guarantee is tight up to constant factors (and remains tight even if one is willing to lose an $\varepsilon$ fraction), due to a lower bound of~\cite{FeldmanFR18}. \\

\noindent\textbf{``Big $n$ Regime'':} For $n$ additive bidders with $m = o(n)$ independent items, Eden et al.~\cite{EdenFFTW17a} prove a competition complexity bound of $n+2(m-1)$. Feldman et al.~\cite{FeldmanFR18} prove that {for any $\varepsilon$, there exists a constant $\delta(\varepsilon)$ such that if $n\geq m/\delta(\varepsilon)$, selling separately (without any additional bidders) achieves a $(1-\varepsilon)$ fraction of the optimal revenue.}
Our main result improves the guarantee of~\cite{EdenFFTW17a} to $9\sqrt{nm}$ and also implies the result of~\cite{FeldmanFR18} (with $\delta(\varepsilon) = \varepsilon^2/81)$. Note in particular that any sublinear competition complexity bound implies the~\cite{FeldmanFR18} result for a different $\delta(\cdot)$, but that linear bounds do not. So our improvement from linear to sublinear is significant in this regard. Moreover, we show in Section~\ref{sec:lbs} that this is tight (up to constant factors) for any approach starting from the benchmark proposed in~\cite{EdenFFTW17a}. We further show (also in Section~\ref{sec:lbs}) that there does not exist any function only of $m$ upper bounding the competition complexity: as $n\rightarrow \infty$ the competition complexity approaches $\infty$ as well (at a rate of at least $\Omega(\ln n)$). 

Other works that study the competition complexity of auctions include seminal work of Bulow and Klemperer, who study the $m=1$ case, work of Liu and Psomas (who study the competition complexity of dynamic auctions) and Roughgarden et al. (who study the unit-demand setting)~\cite{BulowK96, LiuP18, RoughgardenTY12}. These works are thematically related, but both the results and techniques have little overlap. 

Some of the aforementioned works which prove approximation guarantees for simple mechanisms use similar techniques to derive a tractable benchmark that upper bound on the achievable revenue~\cite{ChawlaHK07, ChawlaHMS10, ChawlaMS15, HartN17, LiY13, BabaioffILW14, Yao15, RubinsteinW15, ChawlaM16, CaiZ17, EdenFFTW17b}. However, it is worth noting that all of these works proceed by immediately splitting the benchmark into multiple simpler terms {and finding approximately optimal mechanisms to cover each term separately. The best of those mechanism guarantees approximate optimality to revenue.} 
This greatly simplifies analysis, at the cost of an additional constant factor. Because competition complexity results target the full original revenue, losing this initial constant factor can make future analysis impossible. As a result, while benchmarks may be shared by these lines of work, analysis of the benchmarks is often quite different.

Finally, it is worth noting that recent work follows two approaches to derive revenue upper bounds in these works. Some (including this paper) use virtual value theory~\cite{ChawlaHK07, ChawlaHMS10, RoughgardenTY12, ChawlaMS15, CaiDW16, CaiZ17, EdenFFTW17a, EdenFFTW17b, LiuP18, FuLLT18}. Others use a more direct probabilitistic approach~\cite{HartN17, LiY13, BabaioffILW14, Yao15, RubinsteinW15, ChawlaM16, BabaioffGN17, FeldmanFR18}. For the most part, similar approximation guarantees are achievable through both approaches. With respect to these lines of work, our results (which yield exact competition complexity bounds) in comparison to those of~\cite{FeldmanFR18} (which lose an arbitrarily small $\varepsilon$) suggest the virtual value approach may be desirable if one cares about small losses.

\subsection{Roadmap}
Our main result tightly characterizes the competition complexity in the litte $n$ regime, and tightly characterizes the competition complexity in the big $n$ regime among proofs which use the same benchmark as~\cite{EdenFFTW17a}. 

In Section~\ref{sec:prelim}, we provide the necessary preliminaries surrounding the benchmark of~\cite{EdenFFTW17a} and virtual value theory. In Section~\ref{sec:single} we provide a near-complete proof of our results when $n=1$ as a warm-up. In Section~\ref{sec:multi}, we analyze the benchmark and reduce the analysis to proving stochastic dominance of certain correlated random variables drawn from $[0,1]$. In Section~\ref{sec:dominance} we prove the required claims regarding stochastic dominance (which at this point are purely mathematical claims and no longer reference auctions). In Appendix~\ref{sec:lbs} we: (a) recap the lower bound of~\cite{FeldmanFR18} witnessing that our results are tight in the little $n$ regime, (b) provide a lower bound witnessing that our results are tight in the big $n$ regime (among proofs which use the same benchmark as~\cite{EdenFFTW17a}), and (c) prove that the competition complexity of $n$ bidders with additive valuations over $m$ independent items approaches $\infty$ as $n \rightarrow \infty$.

\section{Notation and Preliminaries}\label{sec:prelim}
We consider a setting with $n$ i.i.d. bidders with additive valuations over $m$ independent items. That is, there are single-variate distributions $D_j$ for all $j \in [m]$, and bidder $i$'s value $v_j$ for item $j$ is drawn independently from $D_j$. Bidder $i$'s value for the bundle $S$ is just $\sum_{j \in S} v_j$. We will use the following notation:
\begin{itemize}
\item $\rev_n(D)$ to denote the revenue of the optimal (possibly randomized) Bayesian Incentive Compatible\footnote{A mechanism is Bayesian Incentive Compatible if it is in every bidder's interest to bid truthfully, conditioned on all other bidders bidding truthfully as well. That is, assuming that all other bidders submit bids drawn from $D_{-i}$, bidder $i$ best responds by bidding their true values.} mechanism when played by $n$ bidders whose values for $m$ items are drawn from $D$. In our setting, we will always have $D = \times_j D_j$ for some single-variate distributions $D_j$.
\item $\vcg_n(D)$ to denote the revenue achieved by the VCG mechanism when played by $n$ bidders whose values for $m$ items are drawn from $D$. In our setting, the VCG mechanism simply runs a second-price auction on each item separately with no reserve.
\item $\srev_n(D)$ to denote the revenue achieved by Myerson's mechanism run separately on each item, when played by $n$ bidders whose values for $m$ items are drawn from $D$. Note that for all $n$ and distributions $D$ over additive valuations, $\srev_n(D) \geq \vcg_n(D)$. 
\end{itemize}

\subsection{Myerson's Lemma, Bulow-Klemperer, and Virtual Values}
Here, we briefly recap basic facts about the theory of virtual values due to Myerson~\cite{Myerson81}. We include some proofs and sketches in Appendix~\ref{app:prelim}, and refer the reader to~\cite{HartlineBook} (Definition 3.11) for a deeper treatment of these concepts (or~\cite{CaiDW16}, Definition 8 for discrete distributions). Note that much of the theory extends to independent (but non-i.i.d.) bidders with slightly more complex statements. As we only consider i.i.d. bidders, we omit the extra notation. Below, when we write $X^+$ for a random variable $X$, we mean $\max\{X,0\}$.

\begin{definition}[Virtual Values and Ironing~\cite{Myerson81}] For a continuous single-variate distribution with CDF $F(\cdot)$ and PDF $f(\cdot)$, the virtual valuation function $\varphi_F(\cdot)$ satisfies $\varphi_F(v) = v - \frac{1-F(v)}{f(v)}$. If $\varphi_F(\cdot)$ is monotone non-decreasing, $F$ is said to be \emph{regular}. If not, $\overline{\varphi}_F(\cdot)$ is the \emph{ironed virtual value} function, and is monotone non-deceasing (see~\cite{HartlineBook} for a formal definition). When $F$ is regular, $\overline{\varphi}_F(\cdot) = \varphi_F(\cdot)$. 
\end{definition}

\begin{theorem}[\cite{Myerson81}]\label{thm:Myerson} Let $D$ be any single-variate distribution. Then for all $n$:
$$\srev_n(D) = \rev_n(D) = \mathbb{E}_{\vec{v} \leftarrow D^n} \left[\left(\overline{\varphi}_D\left(\max_{i \in [n]}\{v_i\}\right)\right)^+\right],\quad \vcg_n(D) = \mathbb{E}_{\vec{v} \leftarrow D^n} \left[\left({\varphi}_D\left(\max_{i \in [n]}\{v_i\}\right)\right)\right].$$
\end{theorem}

\begin{fact}\label{fact:one} For any single-variate distribution $D$, and any value $v$, let $D_{\geq v}$ denote the distribution $D$ conditioned on exceeding $v$. Then $v=\mathbb{E}_{w\leftarrow D_{\geq v}}[\varphi(w)]\leq \mathbb{E}_{w \leftarrow D_{\geq v}}[\overline{\varphi}(w)]$. 
\end{fact}

Finally, we recall the seminal result of Bulow and Klemperer~\cite{BulowK96}:

\begin{theorem}[\cite{BulowK96}] For any regular single-variate distribution $D$, $\vcg_{n+1}(D) \geq \rev_n(D)$. 
\end{theorem}

\subsection{Duality Benchmarks}\label{sec:duality}
Here we state an upper bound on $\rev_n(D)$ when $D$ is additive over independent items. The bound is derived using the duality framework of Cai et al.~\cite{CaiDW16}, and first used by Eden et al.~\cite{EdenFFTW17a} (it is also possible to derive this particular bound without duality~\cite{ChawlaMS15}). When referring to this benchmark in text, we call it the EFFTW benchmark. Parsing the benchmark requires additional notation:

\begin{itemize}
\item $v_{ij}$ denotes the value of bidder $i$ for item $j$.
\item $D_j$ denotes the marginal distribution of item $j$. We use $\overline{\varphi}_j(\cdot)$ to denote $\overline{\varphi}_{D_j}(\cdot)$. 
\item For a variable $X$, if $X$ has no point-masses, then we simply define $F(x) = \Pr[X < x] = \Pr[X \leq x]$. If $X=x$ with strictly positive probability, then we define $F(x)$ to be a random variable drawn uniformly from the interval $[\Pr[X < x], \Pr[X \leq x]]$. Importantly, note that the random variable $F(X)$ is drawn uniformly from $[0,1]$ for any random variable $X$. 
\item For a distribution $D:= \times_j D_j$, we partition the space $\mathbb{R}_+^m$ into $m$ disjoint \emph{regions}. For each $j \in [m]$, we define $R_j:= \{\vec{v}_i \in \mathbb{R}_+^m\ |\ F_j(v_{ij}) > F_{k}(v_{ik})\ \forall k\neq j\}$. That is, $\vec{v}_i$ is in region $R_j$ if item $j$ has the highest \emph{quantile}. Observe that his partition may be randomized if $D$ has point masses (and is deterministic with probability $1$ if $D$ has no point masses). 
\end{itemize}

\begin{theorem}[\cite{CaiDW16,EdenFFTW17a}]\label{thm:CDW} Let $D$ be additive over $m$ independent items. Then:

$$\rev_n(D) \leq \sum_{j=1}^m \mathbb{E}_{\vec{v} \leftarrow D^n} \left[\max_{i \in [n]} \left\{\overline{\varphi}_j(v_{ij})^+\cdot \mathbb{I}(\vec{v}_i \in R_j) + v_{ij} \cdot \mathbb{I}(\vec{v}_i \notin R_j)\right\}\right].$$
\end{theorem}

If we think of the \emph{Virtual Value} of bidder $i$ for item $j$ as equal to Myerson's ironed virtual value, $\overline{\varphi}_j(v_{ij})^+$, if item $j$ has the highest quantile in $\vec{v}_i$, and equal to the value, $v_{ij}$, if not, then Theorem~\ref{thm:CDW} claims that the expected revenue of the optimal mechanism does not exceed the sum over all items of the expected maximum virtual value for that item. Theorem~\ref{thm:CDW} is an application of Corollary~28 in~\cite{CaiDW16}, together with the observation that our defined regions are upwards-closed.

\section{Warm-Up: Single Bidder}\label{sec:single}
In this section, we illustrate one portion of our improved anlaysis via the single bidder setting. This will also help identify one significant point of departure from~\cite{EdenFFTW17a}. Observe that the EFFTW benchmark simplifies significantly for a single bidder, as there is only one element of $[n]$, and the benchmark simply sums the virtual value of the item with the highest quantile plus the values of all other items.

\subsection{Brief Recap of~\cite{EdenFFTW17a}}
The main idea in the single-bidder approach of~\cite{EdenFFTW17a} is to couple draws of $m$ bidders for item $j$ with draws of a single bidder for $m$ items via their quantiles. Specifically, they observe the following: consider fixed quantiles $q_1,\ldots, q_m$ drawn independently. and uniformly from $[0,1]$. 
\begin{itemize}
\item \textbf{Benchmark Analysis:} Use the quantiles drawn to determine values for each of $m$ items. If $q_j$ is the largest quantile drawn, then item $j$ contributes $\overline{\varphi}_j(F^{-1}_j(q_j))^+$ to the benchmark. If $q_j$ is not the largest quantile drawn, then item $j$ contributes $F^{-1}_j(q_j)$ to the benchmark.
\item \textbf{VCG Analysis:} Use the quantiles drawn to determine values of each of $m$ bidders for item $j$. If $q_j$ is the largest quantile drawn, then bidder $j$ contributes $\overline{\varphi}_j(F^{-1}_j(q_j))$ to the virtual surplus of VCG. If $q_j$ is not the largest quantile drawn, then some other bidder wins the item and pays at least $F^{-1}(q_j)$, so at least $F^{-1}(q_j)$ is contributed by some bidder $\neq j$ to the revenue. 
\end{itemize}

The above reasoning is not far from a formal proof that $\srev_m(D) \geq \rev_1(D)$. Some care is required to make sure Theorem~\ref{thm:Myerson} is applied correctly (since we wish to count bidder $j$'s contribution to the revenue of VCG using her ironed virtual value but the other bidders' contributions directly via payments), but the above reasoning is the key step. The main idea is that if we couple the quantiles drawn for the benchmark with quantiles drawn for selling just item $j$, then the revenue achieved from selling just item $j$ to $m$ bidders drawn from $D_j$ exceeds the contribution of item $j$ to the benchmark \emph{for all quantiles drawn}. 

\subsection{Our Analysis}
The main challenge that the previous analysis overcomes is the following: the contribution of item $j$ to the benchmark is sometimes in the form of a virtual value, and sometimes in the form of a value. There is no ``natural'' random variable that takes exactly this form, and it is tricky to analyze directly. So the previous analysis finds a clever way to ``recreate'' it using this coupling argument. Unfortunately though, direct coupling arguments like this should not hope to prove a competition complexity better than $m-1$, as there are $m$ random variables that need to be coupled.

Our approach instead is to reason about the contribution of item $j$ to the benchmark \emph{exclusively in terms of virtual values}, using Fact~\ref{fact:one}. Specifically, consider the following proposition, which rewrites the contribution of item $j$ to the benchmark. Below, $X_L(1,m)$ denotes the following random variable: first, draw one quantile $X_{1,1}$ uniformly at random from $[0,1]$. Then, draw $m-1$ quantiles uniformly at random from $[0,1]$ and label them $Y_{1,m-1}$ thru $Y_{m-1,m-1}$. If $X_{1,1} > Y_{\ell,m-1}$ for all $\ell$, then set $X_L(1,m) = X_{1,1}$. Otherwise, let $\ell^*$ denote a uniformly random element from $\{\ell\ |\ Y_{\ell,m-1} > X_{1,1}\}$ and set $X_L(1,m) = Y_{\ell^*,m-1}$. 

\begin{proposition}\label{prop:singlemain} For all $D = \times_j D_j$ and all items $j$, $\mathbb{E}_{\vec{v}\leftarrow D} \left[\overline{\varphi}_j(v_{j})^+\cdot \mathbb{I}(\vec{v}\in R_j) +v_j\cdot \mathbb{I}(\vec{v} \notin R_j)\right] \leq \mathbb{E}[\overline{\varphi}_j(F^{-1}_j(X_L(1,m)))^+]$.
\end{proposition}

\begin{proof}
The main idea is to get a lot of mileage from Fact~\ref{fact:one}: ideally, any time $\vec{v} \notin R_j$, rather than contribute $v_j$ to the benchmark, we will contribute the virtual value of a random draw from $D_j$ conditioned on exceeding $v_j$. To begin, let's couple quantiles drawn for the benchmark with quantiles drawn for the experiment defining $X_L(1,m)$ so that $X_{1,1} = F_j(v_j)$ and $Y_{\ell,m-1} = F_\ell(v_\ell)$ for $\ell < m, \ell \neq j$, and $Y_{j,m-1} = F_{m}(v_m)$ (if $j \neq m$, otherwise there is no $Y_{m,m-1}$ to define). Observe that indeed the quantiles are all drawn independently and uniformly from $[0,1]$. Moreover, we have:
\begin{itemize}
\item Whenever $\vec{v} \in R_j$, $X_L(1,m) = X_{1,1} = F_j(v_j)$. Therefore, we conclude that: 
\begin{equation}\mathbb{E}_{\vec{v} \leftarrow D}\left[\overline{\varphi}_j(v_j)^+ \cdot \mathbb{I}(\vec{v} \in R_j)\right] = \mathbb{E}\left[\overline{\varphi}_j(F^{-1}_j(X_L(1,m))^+\cdot \mathbb{I}(X_L(1,m) = X_{1,1})\right].\end{equation}
\item Conditioned on $\vec{v} \notin R_j$, $X_L(1,m)$ is a uniformly random sample from $[X_{1,1},1]$. This is because there is some strictly positive number of $\ell$s such that $Y_{\ell,m-1} > X_{1,1}$. Conditioned on being $> X_{1,1}$, each such value is drawn uniformly from $[X_{1,1},1]$. And then $X_L(1,m)$ picks one of them uniformly at random. Using Fact~\ref{fact:one}, we therefore conclude that:
\begin{align}
\mathbb{E}_{\vec{v} \leftarrow D} \left[v_j \cdot \mathbb{I}(\vec{v} \notin R_j)\right] &\leq \mathbb{E}_{\vec{v} \leftarrow D} \left[\mathbb{E}_{x \leftarrow D_{j,\geq v_j}}\left[\overline{\varphi}_j(x)\right] \cdot \mathbb{I}(\vec{v} \notin R_j) \right]\nonumber \\&\leq \mathbb{E}\left[\overline{\varphi}_j(F_j^{-1}(X_L(1,m))) \cdot \mathbb{I}(X_L(1,m) \neq X_{1,1})\right].
\end{align}
\end{itemize}

It is now easy to see that the left-hand sides of the two equations sum together to yield item $j$'s contribution to the benchmark, while the two right-hand sides sum together to yield (at most) $\mathbb{E}[\overline{\varphi}_j(F^{-1}_j(X_L(1,m)))^+]$, proving the proposition.
\end{proof}

Proposition~\ref{prop:singlemain} gives an upper bound on the contribution of item $j$ to the benchmark written as the expectation of a virtual value of \emph{some} distribution ($F^{-1}(X_L(1,m))$). This is convenient because we can write the revenue achieved by using Myerson's optimal auction for selling item $j$ to $1+c$ bidders as the expectation of a virtual value of another distribution (the maximum of $1+c$ draws from $D_j$). Therefore, if we can relate these two distributions (for instance, by proving that one stochastically dominates the other), we can relate these two expectations. Below, let $X_S(1,c)$ denote the maximum of $1+c$ i.i.d. draws from the uniform distribution on $[0,1]$. 

\begin{corollary}\label{cor:reducesingle} If $X_S(1,c)$ stochastically dominates $X_L(1,m)$, then for all $D$ that are additive over $m$ independent items, $\srev_{1+c}(D) \geq \rev_{1}(D)$. 
\end{corollary}
\begin{proof}
Observe first that by Theorem~\ref{thm:Myerson} we have:
$$\srev_{1+c}(D) = \sum_j \mathbb{E}_{\vec{x} \leftarrow D_j^{1+c}}  \left[\overline{\varphi}_j\left(\max_{i \in [1+c]}\{x_i\}\right)^+\right] = \sum_j \mathbb{E}\left[\overline{\varphi}_j(F_j^{-1}(X_S(1,c))^+\right].$$

By Proposition~\ref{prop:singlemain} (and Theorem~\ref{thm:CDW}), we have:

$$\rev_1(D) \leq \sum_{j=1}^m \mathbb{E}_{\vec{v} \leftarrow D} \left[\overline{\varphi}_j(v_{j})^+\cdot \mathbb{I}(\vec{v}\in R_j) + v_{j} \cdot \mathbb{I}(\vec{v} \notin R_j)\right] \leq \sum_{j=1}^m \mathbb{E}\left[\overline{\varphi}_j(F^{-1}_j(X_L(1,m)))^+\right].$$

Observe that $\overline{\varphi}_j(\cdot)$ is a monotone non-decreasing function, and $F_j^{-1}$ is also monotone non-decreasing. As such, if $X_S(1,c)$ stochastically dominates $X_L(1,m)$, $\overline{\varphi}_j(F_j^{-1}(X_S(1,c)))$ stochastically dominates $\overline{\varphi}_j(F_j^{-1}(X_L(1,m)))$, which allows us to conclude that $\sum_j \mathbb{E}\left[\overline{\varphi}_j(F_j^{-1}(X_S(1,c))^+\right] \geq \sum_{j=1}^m \mathbb{E}\left[\overline{\varphi}_j(F^{-1}_j(X_L(1,m)))^+\right]$. Therefore, we may conclude that if $X_S(1,c)$ stochastically dominates $X_L(1,m)$, $\srev_{1+c}(D) \geq \rev_1(D)$. 
\end{proof}

At this point, we've reduced the problem of deriving competition complexity upper bounds to a purely mathematical problem relating stochastic dominance of $X_S(1,c)$ and $X_L(1,m)$. The proof of this claim for $n=1$ is not an especially instructive special case, so we defer the final step to Section~\ref{sec:dominance}. So we wrap up our warm-up by citing Theorem~\ref{thm:littlen}:

\begin{corollary}[of Theorem~\ref{thm:littlen}]\label{cor:lastsingle} When $c \geq 2+\ln(m+1)$, $X_S(1,c)$ stochastically dominates $X_L(1,m)$. 
\end{corollary}

\begin{theorem}\label{thm:single} Let $D$ be a distribution that is additive over $m$ independent items. Then $\srev_{2+\ln(m+1)}(D) \linebreak[0] \geq \rev_{1}(D)$. If each $D_j$ is regular, then also $\vcg_{3+\ln(m+1)}(D) \geq \rev_1(D)$.
\end{theorem}
\begin{proof}
Theorem~\ref{thm:CDW} upper bounds $\rev_1(D)$ with the EFFTW benchmark. Proposition~\ref{prop:singlemain} further upper bounds the EFFTW benchmark with $\sum_j \mathbb{E}[\overline{\varphi}_j(F^{-1}_j(X_L(1,m)))^+]$, which is the sum over all items of the expected virtual value of a quantile drawn from $X_L(1,m)$. Corollary~\ref{cor:reducesingle} argues that if $X_L(1,m)$ is stochastically dominated by $X_S(1,c)$ (the maximum of $c+1$ i.i.d. draws uniformly from $[0,1]$), then we may replace $X_L(1,m)$ with $X_S(1,c)$ in the upper bound, which is exactly $\srev_{1+c}(D)$. Finally, Corollary~\ref{cor:lastsingle} claims that indeed $X_S(1,c)$ stochastically dominates $X_L(1,m)$ when $c \geq 2+\ln(m+1)$ (and the final $+1$ when each $D_j$ is regular comes from going from $\srev$ to $\vcg$ using Bulow-Klemperer). 
\end{proof}

This concludes our exposition for a single bidder. Above we introduced one new idea which departs from prior work: instead of directly treating the benchmark which involves both values and virtual values, rewrite the benchmark to involve only virtual values and reduce the problem to purely mathematical questions about stochastic dominance of $X_L(1,m)$ and $X_S(1,c)$. 
\section{Multiple Bidders}\label{sec:multi}
In this section, we overview our approach for the general case. The key simplifying feature of the single-bidder case that allowed us to isolate one key idea is that for each item $j$, that item has the highest quantile or it doesn't. In the multi-bidder case, there are multiple bidders, some of whom will have their highest quantile for item $j$, some of whom won't. So we must actually engage with the ``$\max_{i \in [n]}$'' in the benchmark. Our approach will be different depending on whether $n$ is big or little relative to $m$. We begin with the little $n$ case as it is more similar to the single-bidder case.

\subsection{Part One: When $n \leq m$}
Our key step is conceptually similar to Proposition~\ref{prop:singlemain}, but the random variables involved are necessarily more complex. We first make the following observation (also made in~\cite{EdenFFTW17a}). Below, $v_{(\ell)j}$ denotes the $\ell^{th}$ highest value for item $j$ (among all bidders). All omitted proofs can be found in Appendix~\ref{app:multi}.

\begin{observation}\label{obs:EFFTW}
$$\mathbb{E}_{\vec{v} \leftarrow D^n} \left[\max_{i \in [n]} \left\{\overline{\varphi}_j({v_{ij}})^+\cdot \mathbb{I}(\vec{v}_i \in R_j)+{v_{ij}} \cdot \mathbb{I}(\vec{v}_i \notin R_j)\right\}\right]$$
$$\leq \mathbb{E}_{\vec{v} \leftarrow D^n} \left[\max \left\{v_{(1)j}\cdot \mathbb{I}(\vec{v}_{(1)} \notin R_j), \overline{\varphi}_j(v_{(1)j}), v_{(2)j} \right\} \right].$$
\end{observation}

Next, we want to rewrite the right-hand side above using random variables simliar to $X_L(1,m)$.
This time, let $X'_L(n,m)$ denote the following random variable: first, draw $n$ quantiles $X_{1,n},\ldots, X_{n,n}$ independently and uniformly at random from $[0,1]$. Relabel them so that $X_{(1),n} \geq \ldots \geq X_{(n),n}$. Then, draw $m-1$ quantiles uniformly at random from $[0,1]$ and label them $Y_{1,m-1}$ thru $Y_{m-1,m-1}$. If $X_{(1),n} > Y_{\ell, m-1}$ for all $\ell$, then set $X'_L(n,m) = X_{(1),n}$. Otherwise, let $\ell^*$ be a uniformly random element from $\{\ell | Y_{\ell, m-1} > X_{(1),n}\}$ and set $X'_L(n,m) = Y_{\ell^*, m-1}$. 
\begin{proposition}\label{prop:littlenmain} For all $D = \times_j D_j$, and all items $j$:
$$\mathbb{E}_{\vec{v} \leftarrow D^n} \left[\max \left\{v_{(1)j}\cdot \mathbb{I}(\vec{v}_{(1)} \notin R_j), \overline{\varphi}_j(v_{(1)j}), v_{(2)j} \right\} \right] \leq \mathbb{E}\left[\max\left\{\overline{\varphi}_j(F_j^{-1}(X'_L(n,m))),F_j^{-1}(X_{(2),n})\right\}\right].$$
\end{proposition}

Proposition~\ref{prop:littlenmain} helps us replace any instances of $v_{(1)n}$ in the benchmark with a randomly drawn virtual value, but we still need to do the same for $v_{(2)n}$ (which so far has essentially just been rewritten as $F_j^{-1}(F_j(v_{(2)n}))$). Now, let $W_{2,n}$ be a uniformly random draw from $[X_{(2), n},1]$, and define $X_L(n,m) = \max\{X'_L(n,m), W_{2,n}\}$. By making use of Fact~\ref{fact:one}, we can conclude:

\begin{corollary}\label{cor:littlenmain}
$$\mathbb{E}\left[\max\left\{\overline{\varphi}_j(F_j^{-1}(X'_L(n,m))),F_j^{-1}(X_{(2),n})\right\}\right] \leq \mathbb{E}\left[\overline{\varphi}_j(F_j^{-1}(X_L(n,m)))\right].$$
\end{corollary}

Now, we are nearly ready to wrap up the $n \leq m$ case. Similarly to the single-bidder case, define $X_S(n,c)$ to be the maximum of $n+c$ i.i.d. draws uniformly form $[0,1]$. 

\begin{corollary}\label{cor:single}If $X_S(n,c)$ stochastically dominates $X_L(n,m)$, then $\srev_{n+c}(D) \geq \rev_n(D)$. If each $D_j$ is regular, then $\vcg_{n+c}(D) \geq \rev_n(D)$.
\end{corollary}

Finally, Theorem~\ref{thm:littlen} claims that when $c \geq n \cdot (2+\ln(1+m/n))$, $X_S(n,c)$ indeed stochastically dominates $X_L(n,m)$. Combining Corollary~\ref{cor:single} with Theorem~\ref{thm:littlen} therefore concludes:

\begin{theorem}\label{thm:littlenmain} For all $D$ that are additive over $m$ independent items, $\srev_{n+n\cdot (2+\ln(1+m/n))}(D) \geq \rev_n(D)$. If each marginal of $D_j$ is regular, then $\vcg_{n+n\cdot (2+\ln(1+m/n))}(D) \geq \rev_n(D)$.
\end{theorem}
When $n \leq m$, this is tight up to constant factors, due to a lower bound of~\cite{FeldmanFR18} (see Appendix~\ref{sec:lbs} for the construction). But when $n \geq m$, this is still linear in $n$. We therefore provide an alternative argument in the following section which achieves the optimal (up to constant factors) competition complexity \emph{that is achievable starting from the EFFTW benchmark} of $\Theta(\sqrt{nm})$. 

\subsection{Part Two: When $n \geq m$}
At a high level, the main difference between how we should analyze the $n \leq m$ case and the $n \geq m$ is as follows: Observation~\ref{obs:EFFTW} immediately upper bounds the EFFTW by upper bounding $\overline{\varphi}_j(v_{(2)j})^+\cdot \mathbb{I}(\vec{v}_{(2)} \in R_j) + v_{(2)j} \cdot \mathbb{I}(\vec{v}_{(2)} \notin R_j)$ with $v_{(2)j}$. When $n \leq m$, this upper bound is unlikely to be much of a relaxation, because it's likely that $v_{(1)j} \notin R_j$ anyway. But when $n \gg m$, we're extremely unlikely to have $v_{(1)j} \notin R_j$, and this upper bound is wasteful. Indeed, this step is what limits the analysis in~\cite{EdenFFTW17a} to $\Omega(n)$. The first step for the $n \geq m$ case is to address this.

\begin{proposition}\label{prop:stepone}
For all items $j$, all $\ell \in [n]$, and all distributions $D$ that are additive over independent items:
$$\mathbb{E}_{\vec{v} \leftarrow D^n} \left[\max_{i \in [n]}\left\{\overline{\varphi}_j(v_{ij})^+ \cdot \mathbb{I}(\vec{v}_i \in R_j) + v_{ij} \cdot \mathbb{I}(\vec{v}_i \notin R_j)\right\} \right] \leq \mathbb{E}_{\vec{w} \leftarrow D_j^{n+(m-1)(\ell-1)}} \left[\max \left\{\overline{\varphi}_j(w_{(1)}), w_{(\ell)}\right\}\right].$$
\end{proposition}

We now want to take a simliar step to the previous case and replace $w_{(\ell)}$ with a randomly drawn virtual value using Fact~\ref{fact:one}. Here, define the random variable $X_B(n,\ell)$ as follows. First, draw $X_{1,n},\ldots, X_{n,n}$ independently and uniformly at random from $[0,1]$. Then, randomly draw $W_{\ell,n}$ uniformly from $[X_{(\ell),n},1]$, and set $X_B(n,\ell):= \max\{X_{(1),n},W_{\ell,n}\}$. 

\begin{lemma}\label{lem:w} For any single-dimensional distribution $D$, and any $n'$:
$$\mathbb{E}_{\vec{w} \leftarrow D^{n'}} \left[\max \left\{\overline{\varphi}_j(w_{(1)}), w_{(\ell)}\right\}\right] \leq \mathbb{E}\left[\overline{\varphi}_j(F_j^{-1}(X_B(n',\ell)))\right].$$
\end{lemma}

Corollary~\ref{cor:bignfinal} below follows from Proposition~\ref{prop:stepone} and Lemma~\ref{lem:w} with $n':= n+(m-1)(\ell-1)$. 

\begin{corollary}\label{cor:bignfinal} If $X_S(n,c)$ stochastically dominates $X_B(n+(\ell-1)(m-1),\ell)$, then for any $D$ that is additive over $m$ independent items, $\srev_{n+c}(D) \geq \rev_n(D)$. If each marginal $D_j$ is regular, then $\vcg_{n+c}(D) \geq \rev_n(D)$. 
\end{corollary}

Finally, Theorem~\ref{thm:bign} states that $X_S(n,c) = X_S(n+(\ell-1)(m-1),c-(\ell-1)(m-1))$ stochastically dominates $X_B(n+(\ell-1)(m-1),\ell)$ whenever $c-(\ell-1)(m-1) \geq \frac{4n + 4(\ell-1)(m-1)}{\ell-1}$. Setting $\ell:= \sqrt{n}{m}+1$, we get $c \geq \sqrt{nm} + 4\sqrt{nm} + 4(m-1)$. 

\begin{theorem}\label{thm:bignmain}
For all $D$ that are additive over $m$ independent items, $\srev_{n+5\sqrt{nm}+4(m-1)}(D) \geq \rev_n(D)$. If each marginal $D_j$ is regular, then $\vcg_{n+5\sqrt{nm} + 4(m-1)}(D) \geq \rev_n(D)$. In particular, if $n \geq m$, $5\sqrt{nm}+4(m-1) \leq 9\sqrt{nm}$. 
\end{theorem}
\section{Stochastic Dominance via Additional Samples}\label{sec:dominance}
In this section, we consider purely questions about whether one distribution stochastically dominates another ({Sections~\ref{sec:single} and~\ref{sec:multi} outline} the connection between these problems and our main result). Recall the following ingredients in our experiments:
\begin{itemize}
\item $X_{1,n},\ldots, X_{n,n}$ are $n$ i.i.d. draws from the uniform distribution on $[0,1]$, and then relabeled so that $X_{(1),n} \geq \ldots \geq X_{(n),n}$. 
\item $Y_{1,m-1},\ldots, Y_{m-1,m-1}$ are $m-1$ i.i.d. draws from the uniform distribution on $[0,1]$, and then relabeled so that $Y_{(1),m-1} \geq \ldots \geq Y_{(m-1),m-1}$. 
\item $Z_{1,c},\ldots, Z_{c,c}$ are $c$ i.i.d. draws from the uniform distribution on $[0,1]$, and then relabeled so that $Z_{(1),c} \geq \ldots \geq Z_{(c),c}$.
\item $W_{\ell,n}$ is a random draw from the uniform distribution on $[X_{(\ell),n},1]$. 
\end{itemize}

\noindent \textsc{SRev Experiment($n,c$):} Output $\sre:= \max \{X_{(1),n},Z_{(1),c}\}$.\\

\noindent \textsc{Big $n$ Benchmark Experiment($n,\ell$):} Output $\bbm:= \max\{X_{(1),n},W_{\ell,n}\}$. \\

\noindent \textsc{Little $n$ Benchmark Experiment($n, m$):} Let $j^*$ be the largest index such that $Y_{(j^*),m-1} > X_{(1),n}$ (if such a $j^*$ exists). If no such $j^*$ exists, output $\lbm:= \max\{X_{(1),n},W_{2,n}\}$. Otherwise, pick an index $j$ uniformly at random from $\{1,\ldots, j^*\}$ and output $\max\{Y_{(j),m-1},W_{2,n}\}$. \\

The main results of this section are as follows:

\begin{theorem}\label{thm:bign} When $c \geq 4n/(\ell-1)$, $\sre$ stochastically dominates $\bbm$.
\end{theorem}

\begin{theorem}\label{thm:littlen} When $c \geq n\cdot(2+\ln(1+m/n))$, $\sre$ stochastically dominates $\lbm$.
\end{theorem}

Intuitively, we might expect $X_S(n,c)$ to stochastically dominate $X_B(n,\ell)$ right around $c = 2n/\ell$. This is because $\mathbb{E}[Z_{(1),c}] = 1-1/(c+1)$,
and $\mathbb{E}[W_{\ell,n}] = 1-\frac{\ell}{2(n+1)}$. Of course, this observation doesn't come close to proving stochastic dominance, especially because $X_{(1),n}$ and $W_{\ell,n}$ aren't independent. But it does give us an idea of the right ballpark to shoot for. The following proposition will be used in the proof of both theorems.

\begin{proposition}\label{prop:key}Let $c \geq 4n/(\ell-1)$. Then for all $p$, $\Pr[Z_{(1),c} > p | X_{(1),n} < p] \geq \Pr[W_{\ell,n} > p | X_{(1),n} < p]$. When $\ell =2$, this can be improved to $c \geq n$. 
\end{proposition}
Before getting into the proof, let's unpack the role of conditioning on $X_{(1),n}$. $Z_{(1),c}$ and $X_{(1),n}$ are independent, so $\Pr[Z_{(1),c} > p | X_{(1),n} < p] = \Pr[Z_{(1),c} > p]$. On the other hand, $W_{\ell,n}$ and $X_{(1),n}$ are \emph{positively} correlated: the lower bound on the range from which $W_{\ell,n}$ is drawn is $X_{(\ell),n}$, which is positively correlated with $X_{(1),n}$. So certainly if we could prove the lemma without the conditioning on $X_{(1),n} < p$, the desired proposition would hold. This approach works for $\ell=2$ (and indeed shows up in our proof as a base case), but without conditioning the conclusion is otherwise false for larger $\ell$. 

\begin{proof}
The proof will proceed by induction on $n, \ell$. We begin with the base case, $\ell = 2$. $Z_{(1),c}$ is easy to reason about: $Z_{(1),c}$ is just the maximum of $c$ i.i.d. draws uniformly from $[0,1]$. So:
\begin{equation}\label{eq:z1} \Pr[Z_{(1),c} > p | X_{(1),n} < p] = \Pr[Z_{(1),c} > p] = 1-p^c.
\end{equation}

Now we turn to $W_{2,n}$. As referenced in the foreword to the proof, for this case the proposition statement holds even without conditioning on $X_{(1),n} < p$. Indeed, observe that without conditioning on $X_{(1),n} < p$, $X_{(2),n}$ is just the second-highest of $n$ i.i.d. draws uniformly from $[0,1]$, and $W_{2,n}$ is drawn uniformly from $[0,1]$, but conditioned on exceeding $X_{(2),n}$. That is, $W_{2,n}$ is actually identically distributed to $X_{(1),n}$, and is distributed according to the maximum of $n$ i.i.d. draws uniformly from $[0,1]$. Therefore, when $c = n$, $W_{2,n}$ is identically distributed to $Z_{(1),c}$, and the conclusion holds. That is:
\begin{equation} \Pr[W_{2,n} > p | X_{(1),n} < p] \leq \Pr[W_{2,n} > p] = 1-p^n.
\end{equation}

As such, we have proved the base case (in fact, a slightly stronger claim): for all $n$, and $\ell = 2$ when $c \geq n=n/(\ell-1)$, $Z_{(1),c}$ stochastically dominates $W_{2,n}$. Now we turn to the inductive step, which is significantly more involved. As referenced in the foreword, we must take a different approach for larger $\ell$, as the desired claim is false if we remove conditioning on $X_{(1),n} < p$. 

To this end, we'll first observe that when $p=1$, $\Pr[Z_{(1),c} > 1] = \Pr[W_{\ell,n} > 1 | X_{(1),n} < 1] = 0$, and when $p \rightarrow 0$, $\Pr[Z_{(1),c} > p] = \Pr[W_{\ell,n} > p | X_{(1),n} < p] = 1$. So the desired inequalities hold at both endpoints of $[0,1]$, and we'd like to reason about $p \in (0,1)$. To accomplish this, it will actually be easier to compare $\Pr[Z_{(1),c} > p]\cdot \Pr[X_{(1),n} < p]$ to $\Pr[W_{\ell,n} > p \wedge X_{(1),n} < p]$ (observe that this simply multiplies both conditional probabilities in our original comparison by $\Pr[X_{(1),n} < p]$), and consider the derivative with respect to $p$.

So let $f_{1,n}(\cdot)$ denote the density of $X_{(1),n}$. Then $\Pr[W_{\ell,n} > p \wedge X_{(1),n} < p] = \int_0^p f_{1,n}(q) \cdot \Pr[W_{\ell,n} > p | X_{(1),n} = q] dq$. By Leibniz' rule, the derivative of this with respect to $p$ is:
\begin{align*}
&\frac{\partial \Pr[W_{\ell,n} > p \wedge X_{(1),n} < p]}{\partial p}\\
 =& \frac{ \partial \int_0^p f_{1,n}(q) \cdot  \Pr[W_{\ell,n} > p | X_{(1),n} = q]dq}{\partial p} \\
=& f_{1,n}(p) \cdot \Pr[W_{\ell,n} > p | X_{(1),n} = p] + \int_0^p f_{1,n}(q) \cdot \frac{\partial \Pr[W_{\ell,n} > p | X_{(1),n} = q]}{\partial p}dq. \end{align*}

Let's first unpack the left-most term with the following lemma.

\begin{lemma}\label{lem:recurse} For all $\ell, n > 1$, $\Pr[W_{\ell,n} > p | X_{(1),n} = p] = \Pr[W_{\ell-1, n-1} > p | X_{(1),n-1} < p]$. 

\end{lemma} 
\begin{proof}
Observe that, conditioned on $X_{(1),n} = p$, $X_{(2),n},\ldots, X_{(n),n}$ are $n-1$ (sorted) i.i.d. draws uniformly at random from $[0,p]$, and $X_{(\ell),n}$ is the $(\ell-1)^{th}$ highest of them. Put another way, $X_{(\ell),n}$ conditioned on $X_{(1),n} = p$ is \emph{identically distributed} to $X_{(\ell-1),n-1}$ conditioned on $X_{(1),n-1} < p$. This therefore implies that $W_{\ell,n}$ conditioned on $X_{(1),n} = p$ and $W_{\ell-1,n-1}$ conditioned on $X_{(1),n-1} < p$ are identically distributed as well.
\end{proof}

Now we turn to the right-most term.

\begin{lemma}\label{lem:rightterm}
For all $p, \ell, n, q$, $\frac{\partial \Pr[W_{\ell,n}> p | X_{(1),n} = q]}{\partial p} = -\Pr[W_{\ell,n} > p | X_{(1),n} = q]/(1-p)$.
\end{lemma}

\begin{proof} Let's first expand $\Pr[W_{\ell,n} > p | X_{(1),n} = q]$ by letting $f_{\ell,n}^q(\cdot)$ denote the density of $X_{(\ell),n}$ conditioned on $X_{(1),n} = q$. 

$$\Pr[W_{\ell,n} > p | X_{(1),n} = q] = \int_0^q f_{\ell,n}^q(r)\cdot \frac{1-p}{1-r} dr.$$

This is simply because, conditioned on $X_{(\ell),n} = r$, the probability that $W_{\ell,n}$ (a uniformly random draw from $[r,1]$) exceeds $p$ is exactly $\frac{1-p}{1-r}$. Taking now the derivative with respect to $p$ (again by Leibniz' rule), we see that:
$$\frac{\partial \Pr[W_{\ell,n}> p | X_{(1),n} = q]}{\partial p} = - \int_0^q f_{\ell,n}^q(r) /(1-r) dr = -\Pr[W_{\ell,n}>p|X_{(1),n}=q]/(1-p).$$
\end{proof}

Using Lemma~\ref{lem:rightterm}, we can now rewrite:
\begin{align*}
\int_0^p f_{1,n}(q) \cdot \frac{\partial \Pr[W_{\ell,n}> p | X_{(1),n} = q]}{\partial p}dq &= \frac{-1}{1-p}\int_0^p f_{1,n}(q) \cdot \Pr[W_{\ell,n} > p | X_{(1),n} = q]dq\\
&= \frac{-\Pr[W_{\ell,n} > p \wedge X_{(1),n} < p]}{1-p}
\end{align*}

And using both Lemmas~\ref{lem:recurse} and~\ref{lem:rightterm} we can now simplify (the second equality follows by recalling that $f_{1,n}(\cdot)$ is the density of $X_{(1),n}$):
\begin{align}\label{eq:key}\frac{\partial \Pr[W_{\ell,n} > p \wedge X_{(1),n} < p]}{\partial p} &=  f_{1,n}(p) \cdot \Pr[W_{\ell-1,n-1} > p | X_{(1),n-1} < p] - \frac{\Pr[W_{\ell,n}>p \wedge X_{(1),n} < p]}{1-p}\\
&=np^{n-1}\cdot \Pr[W_{\ell-1,n-1}>p|X_{(1),n-1}<p] -\frac{\Pr[W_{\ell,n} > p \wedge X_{(1),n} < p]}{1-p}\end{align}

From here we'll show that whenever $\Pr[W_{\ell,n} > p \wedge X_{(1),n} < p] \geq \Pr[Z_{(1),c} > p \wedge X_{(1),n} < p]$ (i.e. whenever what we're trying to prove at $p$ is violated), then the derivative trends towards satisfying our desired claim.

\begin{lemma}\label{lem:derivs}If $\Pr[W_{\ell,n} > p \wedge X_{(1),n} < p] \geq \Pr[Z_{(1),c} > p \wedge X_{(1),n} < p]$, then:
$$\frac{\partial \Pr[Z_{(1),c}>p \wedge X_{(1),n} < p] - \Pr[W_{\ell,n} > p \wedge X_{(1),n} < p]}{\partial p} \geq 0.$$
\end{lemma}
\begin{proof}
Observe first that $\Pr[Z_{(1),c} > p \wedge X_{(1),n} < p] = p^n \cdot (1-p^c)$. As such, we also have $\frac{\partial \Pr[Z_{(1),c} > p \wedge X_{(1),n} < p]}{\partial p} = np^{n-1}(1-p^c) - cp^{n+c-1}$. So if the hypotheses of the lemma are satisfied, then by Equation~\eqref{eq:key} we can write:
\begin{align*}
&\frac{\partial \Pr[Z_{(1),c}>p \wedge X_{(1),n} < p] - \Pr[W_{\ell,n} > p \wedge X_{(1),n} < p]}{\partial p}\\
&= np^{n-1}(1-p^c) - cp^{n+c-1}-\left(np^{n-1}\cdot \Pr[W_{\ell-1,n-1}>p|X_{(1),n-1}<p] -\frac{\Pr[W_{\ell,n} > p \wedge X_{(1),n} < p]}{1-p}\right)\\
&\geq np^{n-1}(1-p^c) - cp^{n+c-1} - np^{n-1} \cdot (1-p^{4(n-1)/(\ell-2)}) + \frac{p^n(1-p^c)}{1-p}.
\end{align*}

In the inequality, we have used two facts. First, we have used the inductive hypothesis, which claims that $\Pr[W_{\ell-1,n-1} > p | X_{(1),n-1} < p] \leq \Pr[Z_{(1),4(n-1)/(\ell-2)} > p | X_{(1),n-1} < p] = 1-p^{4(n-1)/(\ell-2)}$. 
Second, we have used the hypothesis of the lemma statement. Next, we can substitute $(1-p^c) = (1+p+\ldots +p^{c-1}) \cdot (1-p)$ (and make some other algebraic simplifications) to get:
\begin{align*}
&\frac{\partial \Pr[Z_{(1),c}>p \wedge X_{(1),n} < p] - \Pr[W_{\ell,n} > p \wedge X_{(1),n} < p]}{\partial p}\\
&\geq -(n+c)p^{n+c-1} + np^{n-1+4(n-1)/(\ell-2)}+p^n \cdot \sum_{j=0}^{c-1} p^j.
\end{align*}

Recall again that we are hoping to prove that the above term is $\geq 0$. As $p \geq 0$, the above term is $\geq 0$ if and only if (dividing all terms by $p^{n-1}$ and rearranging):

\begin{equation}\label{eq:convex}
np^{4(n-1)/(\ell-2)} + \sum_{j=1}^c p^j \geq (n+c) p^c.
\end{equation}

To conclude that Equation~\eqref{eq:convex} indeed holds for all $p \in [0,1]$, consider the convex function $f(x) := p^x$, and the random variable $A$ where $A = j$ with probability $1/(n+c)$ for all $j \in \{1,\ldots, c\}$, and $A = 4(n-1)/(\ell-2)$ with probability $n/(n+c)$. Then the left-hand side of the equation above is exactly $(n+c)\cdot \mathbb{E}[f(A)]$. Therefore, we may conclude by Jensen's inequality that the left-hand side exceeds $(n+c)\cdot f(\mathbb{E}[A]) = (n+c)\cdot p^{\frac{4(n-1)n/(\ell-2)+c(c+1)/2}{n+c}}$. As $p \in [0,1]$, $(n+c)\cdot p^{\frac{4(n-1)n/(\ell-2)+c(c+1)/2}{n+c}} \geq (n+c)p^c$ if and only if $\frac{4(n-1)n/(\ell-2)+c(c+1)/2}{n+c} \leq c$. So finally, our only remaining job is to see what values of $c$ satisfy:

\begin{equation}\frac{4(n-1)n}{\ell-2} + \frac{c(c+1)}{2} \leq nc+c^2.
\end{equation}

Indeed, observe that when $c = 4n/(\ell-1)$, we get:
\begin{align*}
\frac{4(n-1)n}{\ell-2} + \frac{c(c+1)}{2}&\stackrel{?}{\leq } nc+c^2\\
\frac{4}{\ell-2} -\frac{4}{n(\ell-2)}+ \frac{4^2}{2(\ell-1)^2} +\frac{4}{2n(\ell-1)} &\stackrel{?}{\leq} \frac{4}{\ell-1} + \frac{4^2}{(\ell-1)^2}\\
2(\ell-1)^2-2(\ell-1)^2/n + 4(\ell-2) +(\ell-2)(\ell-1)/n &\stackrel{?}{\leq}2(\ell-2)(\ell-1) + 2\cdot 4\cdot (\ell-2)\\
2(\ell-1)  -\ell(\ell-1)/n &\stackrel{?}{\leq} 4\cdot (\ell-2)\\
2\frac{\ell-1}{\ell-2} - \frac{\ell(\ell-1)}{(\ell-2)n} &\stackrel{?}{\leq}4\\
\end{align*}
The last inequality indeed holds as $\ell \geq 3$. 

\end{proof}

Now we are ready to wrap up the proof of the proposition. We have just shown (Lemma~\ref{lem:derivs}) that for all $p$, either $\Pr[Z_{(1),c} > p \wedge X_{(1),n} < p] - \Pr[W_{\ell,n} > p \wedge X_{(1),n} < p] \geq 0$, or $\frac{\partial \Pr[Z_{(1),c}>p \wedge X_{(1),n} < p]}{\partial p} - \frac{\partial  \Pr[W_{\ell,n} > p \wedge X_{(1),n} < p]}{\partial p} \geq 0$.
We now are ready to claim that this implies that $\Pr[Z_{(1),c} > p \wedge X_{(1),n} < p] - \Pr[W_{\ell,n} > p \wedge X_{(1),n} < p] \geq 0$ for all $p \in [0,1]$. 

Indeed, define $G(p) := \Pr[Z_{(1),c} > p \wedge X_{(1),n} < p] - \Pr[W_{\ell,n} > p \wedge X_{(1),n} < p]$. Then we have shown that for all $p$, either $G(p) \geq 0$ or $G'(p) \geq 0$. Moreover, we know that $G(0) = 0-0 = 0$. So assume for contradiction that there exists some $p$ with $G(p) < 0$. Then because $G(\cdot)$ is continuous (and $G(0) = 0$), there exists some open interval $(q,p)$ such that $G(x) < 0$ on $(q,p)$, while $G(q) = 0$. But now we have a contradiction: By Lemma~\ref{lem:derivs}, $G'(x) \geq 0$ on $(q,p)$, and $G(q) = 0$. Therefore, we must also have $G(x) \geq 0$ on $(q,p)$, contradicting our initial assumption. 

Therefore, we cannot have $G(x) < 0$ anywhere on $[0,1]$, meaning that $\Pr[Z_{(1),c} > p \wedge X_{(1),n} < p] - \Pr[W_{\ell,n} > p \wedge X_{(1),n} < p] \geq 0$. This is identical to the claim that $\Pr[Z_{(1),c} > p | X_{(1),n} < p] \geq \Pr[W_{\ell,n} > p | X_{(1),n} < p]$. This completes the proof of the proposition.

\end{proof}

Theorem~\ref{thm:bign} now follows nearly directly from Proposition~\ref{prop:key}.

\begin{proof}[Proof of Theorem~\ref{thm:bign}]
We can directly compute $\Pr[X_S(n,c) > p] = \Pr[X_{(1),n} > p] + \Pr[Z_{(1),c} > p | X_{(1),n} < p]\cdot \Pr[X_{(1),n} < p]$. Similarly, $\Pr[X_B(n,\ell) > p] =\Pr[X_{(1),n} > p] + \Pr[W_{\ell,n} > p|X_{(1),n} < p]\cdot \Pr[X_{(1),n} < p]$. By Proposition~\ref{prop:key}, when $c \geq 4n/(\ell-1)$, we get:
$$\Pr[Z_{(1),c} > p | X_{(1),n} < p] \geq \Pr[W_{\ell,n} > p | X_{(1),n} < p].$$
Therefore,
\begin{align*}
[X_{(1),n} > p] &+ \Pr[Z_{(1),c} > p | X_{(1),n} < p]\cdot \Pr[X_{(1),n} < p]\\ \geq &\Pr[X_{(1),n} > p] + \Pr[W_{\ell,n} > p|X_{(1),n} < p]\cdot \Pr[X_{(1),n} < p].
\end{align*}
This implies
$$\Pr[X_S(n,c) > p] \geq \Pr[W_{\ell,n} > p].$$

As the above holds for all $p \in [0,1]$, this proves that $X_S(n,c)$ stochastically dominates $X_B(n,\ell)$.
\end{proof}

Theorem~\ref{thm:littlen} will require one more similar proposition.

\begin{proposition}\label{prop:littlen} Let $c \geq n\cdot(1+\ln(1+m/n))$, then for all $p$, $\Pr[Z_{(1),c} > p | X_{(1),n} < p] \geq \Pr[Y^*_{n,m-1} > p | X_{(1),n} < p]$, where $Y^*_{n,m-1}$ is the random variable equal to $0$ if $Y_{1,m-1} < X_{(1),n}$, and is otherwise equal to $Y_{j,m-1}$ for a uniformly random $j \in \{j | Y_{j,m-1} > X_{(1),n}\}$. 
\end{proposition}
\begin{proof}
The proof of Proposition~\ref{prop:littlen} is more direct than that of Proposition~\ref{prop:key}. This time, we can just directly compute $\Pr[Y^*_{n,m-1}>p | X_{(1),n} < p]$. We again begin by observing that $\Pr[Z_{(1),c} > p | X_{(1),n} < p] = \Pr[Z_{(1),c} > p] = 1-p^c$.

We now turn to $Y^*_{n,m-1}$. Observe first that $Y^*_{n,m-1} = 0$, conditioned on $X_{(1),n} = q$, with probability exactly $q^{m-1}$. This is because $Y^*_{n,m-1}$ is $0$ whenever each of $m-1$ i.i.d. draws uniformly from $[0,1]$ are all $< q$. Now, conditioned on $Y^*_{n,m-1} > 0$ (and also $X_{(1),n} = q$), observe that $Y^*_{n,m-1}$ is just a random draw from the uniform distribution on $[X_{(1),n},1]$. This is because, conditioned on $Y^*_{n,m-1} > 0$ and $X_{(1),n} = q$, $Y^*_{n,m-1}$ simply picks uniformly at random among some number of i.i.d. random variables drawn uniformly from $[X_{(1),n},1]$. Therefore, conditioned on $Y^*_{n,m-1} > 0$, and $X_{(1),n} =q$, $Y^*_{n,m-1}$ exceeds $p$ with probability exactly $\frac{1-p}{1-q}$. Therefore, we can compute (below, let $f_1(\cdot)$ denote the density of $X_{(1),n}$, and $F_1(\cdot)$ denote the CDF):

\begin{align*}
&\Pr[Y^*_{n,m-1} > p | X_{(1),n} < p] \\
=& \int_0^p \frac{f_1(q)}{F_1(p)} \cdot \Pr[Y^*_{n,m-1} > 0 | X_{(1),n} = q] \cdot \Pr[Y^*_{n,m-1} > p | Y^*_{n,m-1} > 0 \wedge X_{(1),n} = q] dq\\
=& \int_0^p \frac{nq^{n-1}}{p^n} \cdot (1-q^{m-1}) \cdot \frac{1-p}{1-q} dq\\
=&\frac{n(1-p)}{p^n} \cdot \int_0^p q^{n-1} \cdot \frac{1-q^{m-1}}{1-q} dq\\
=&\frac{n(1-p)}{p^n} \cdot \int_0^p q^{n-1} \cdot (\sum_{i=0}^{m-2}q^i) dq\\
=&\frac{n(1-p)}{p^n} \cdot \left[ \sum_{i=0}^{m-2}q^{i+n}/(n+i)\right]_0^p\\
=&\frac{n(1-p)}{p^n} \cdot \sum_{i=0}^{m-2}p^{i+n}/(n+i)\\
=& n(1-p) \cdot \sum_{i=0}^{m-2}p^{i}/(n+i)\\
=& 1 - np^{m-1}/(n+m-2) + n\cdot \sum_{i=1}^{m-2} p^i(1/(n+i) - 1/(n+i-1))\\
=& 1 - np^{m-1}/(n+m-2) - \sum_{i=1}^{m-2} np^i/(n+i)(n+i-1)
\end{align*}

Before proceeding, we quickly observe that the sums of the coefficients of non-zero powers of $p$ is $-1$ (that is, $n/(n+m-2) + \sum_{i=1}^{m-2} n/(n+i)(n+i-1) = 1$). This is because the third-from-the-bottom equality is clearly equal to $0$ when $p = 1$, and so the bottom equality must be equal to $0$ when $p=1$ as well.
\begin{align*}
\Pr[Z_{(1),c} > p | X_{(1),n} < p] &\stackrel{?}{\geq} \Pr[Y^*_{n,m-1} > p | X_{(1),n} < p]\\
1-p^c &\stackrel{?}{\geq}1 - np^{m-1}/(n+m-2) - \sum_{i=1}^{m-2} np^i/(n+i)(n+i-1)\\
p^c &\stackrel{?}{\leq}np^{m-1}/(n+m-2)+ \sum_{i=1}^{m-2} np^i/(n+i)(n+i-1)
\end{align*}

From here, we again apply Jensen's inequality. Let $f(x):=p^x$ (which is convex), and let $A$ denote the random variable which is equal to $m-1$ with probability $n/(n+m-2)$ and equal to $i$ with probability $n/(n+i)(n+i-1)$ for all $i \in \{1,\ldots, m-2\}$. By reasoning in the previous paragraph (that the coefficients of non-zero powers of $p$ sum to $-1$), this is indeed a distribution. Then Jensen's inequality asserts that $\mathbb{E}[f(A)] \geq f(\mathbb{E}[A])$. Moreover, the right-hand side above is exactly $\mathbb{E}[f(A)]$. As such, we get that:
$$np^{m-1}/(n+m-2)+ \sum_{i=1}^{m-2} np^i/(n+i)(n+i-1) \geq p^{n(m-1)/(n+m-2) + \sum_{i=1}^{m-2} ni/(n+i)(n+i-1)}.$$

As $p \in [0,1]$, this means that our desired inequality is satisfied as long as $c \geq n(m-1)/(n+m-2) + \sum_{i=1}^{m-2} ni/(n+i)(n+i-1)$. But now observe that:
\begin{align*}
n(m-1)/(n+m-2) + \sum_{i=1}^{m-2} ni/(n+i)(n+i-1)  &= n\left(\frac{m-1}{n+m-2} + \sum_{i=1}^{m-2} \frac{i}{(n+i)(n+i-1)}\right)\\
&\leq n \cdot \left(1+\sum_{i=1}^{m-2} 1/(n+i)\right)\\
&\leq n \cdot \left(1+\ln(\frac{n+m-2}{n})\right)\\
&\leq n \cdot \left(1+\ln(1+\frac{m}{n})\right).
\end{align*}
\end{proof}

And now we can prove Theorem~\ref{thm:littlen}, which essentially combines Proposition~\ref{prop:key} and Proposition~\ref{prop:littlen} (with some extra work).

\begin{proof}[Proof of Theorem~\ref{thm:littlen}] 
We again directly compute $\Pr[X_S(n,c) > p] = \Pr[X_{(1),n} > p] + \Pr[Z_{(1),c} > p | X_{(1),n} < p]$. Similarly, $\Pr[X_L(n,m) > p] = \Pr[X_{(1),n} > p] + \Pr[W_{2,n} > p | X_{(1),n} < p]\cdot \Pr[X_{(1),n} < p] + \Pr[Y^*_{n,m-1} > p | X_{(1),n} < p \wedge W_{2,n} < p]\cdot \Pr[X_{(1),n} < p \wedge W_{2,n} < p]$, where again $Y^*_{n,m-1}$ is defined to be $0$ if $Y_{1,m-1} < X_{(1),n}$, and otherwise equal to $Y_{j,m-1}$ for a uniformly random $j \in \{j|Y_{j,m-1}> X_{(1),n}\}$. 

By Proposition~\ref{prop:key} we have that $\Pr[W_{2,n} > p | X_{(1),n} < p] \leq \Pr[Z_{(1),n} > p]$. Now we need to reason about $Y^*_{n,m-1}$ conditioned on $X_{(1),n} < p$ and $W_{2,n} < p$, which is not directly related to any previous propositions. However, we claim that $Y^*_{n,m-1}$ and $W_{2,n}$ are positively correlated, conditioned on $X_{(1),n} < p$.

\begin{lemma}\label{lem:correlation} $\Pr[Y^*_{n,m-1} > p | X_{(1),n} < p \wedge W_{2,n} < p] \leq \Pr[Y^*_{n,m-1} > p | X_{(1),n} < p]$.
\end{lemma}
\begin{proof}
For this proof, for random variables $A, B, C$, when we say $A$ and $B$ are conditionally independent, conditioned on $C$ we mean that for all $c$, the random variables $A$ and $B$ are conditionally independent, conditioned on the event $C= c$. We use this shorthand to avoid cumbersome notation.

The proof consists of three steps: (a) we first show $Y^*_{n,m-1}$ and $W_{2,n}$ are conditionally independent, conditioned on $X_{(1),n}$, (b) we show that, conditioned on $X_{(1),n}<p$, $Y^*_{n,m-1}$ is positively correlated with $X_{(1),n}$, and (c), $W_{2,n}$ is positively correlated with $X_{(1),n}$. Together, this essentially lets us claim that additionally conditioning on $W_{2,n} < p$ only serves to lower $X_{(1),n}$, which lowers the probability that $Y^*_{n,m-1} > p$.

Observe first that conditioned on $X_{(1),n},\ldots, X_{(n),n}$, $Y^*_{n,m-1}$ and $W_{2,n}$ are independent (this is just by definition: they are drawn independently, but the distributions from which they are drawn depend on $X_{(1),n},\ldots, X_{(n),n}$). However, observe that the distribution from which $Y^*_{n,m-1}$ is independently drawn can be defined as a function only of $X_{(1),n}$. This means that, conditioned on $X_{(1),n}$, $X_{(2),n}$ and $Y^*_{n,m-1}$ are conditionally independent. Similarly, the distribution from which $W_{2,n}$ is independently drawn from can be described as a function only of $X_{(2),n}$, which we just claimed is conditionally independent of $Y^*_{n,m-1}$, conditioned on $X_{(1),n}$. Therefore, $Y^*_{n,m-1}$ and $W_{2,n}$ are conditionally independent, conditioned on $X_{(1),n}$.\footnote{Note that $Y^*_{n,m-1}$ and $W_{2,n}$ are \emph{not} conditionally independent, conditioned on $X_{(1),n} < p$. They are only conditionally independent, conditioned on $X_{(1),n} = q$ (for some $q$). }

Next, we want to claim that, for all $r \leq q \leq p$, $\Pr[Y^*_{n,m-1} > p | X_{(1),n} = q] \geq \Pr[Y^*_{n,m-1} > p | X_{(1),n} = r] $. To see this, observe the following equivalent method for drawing $Y^*_{n,m-1}$: First, draw $Y_{1,m-1},\ldots, Y_{m-1,m-1}$ i.i.d. from the uniform distribution on $[0,1]$. Permute them into random order.\footnote{Actually, this step is not necessary, but it helps the analogy to state it.} Then, let $j$ be the smallest index such that $Y_{j,m-1} > X_{(1),n}$. If no such $j$ exists, set $Y^*_{n,m-1} = 0$. Otherwise, set $Y^*_{n,m-1} = Y_{j,m-1}$. Now, let's couple draws for $Y^*_{n,m-1}$ conditioned on $X_{(1),n} = q$ and $X_{(1),n} = r$ by fixing the values $Y_{1,m-1},\ldots, Y_{m-1,m-1}$ and the random permutation. Then think of $Y^*_{n,m-1}$, conditioned on $X_{(1),n} = r$ (respectively, $X_{(1),n} = q$) as scanning the values sequentially until it hits one whose value exceeds $r$ (respectively, $q$). 
\begin{itemize}
\item If the (permuted) sequence $Y_{1,m-1},\ldots, Y_{m-1,m-1}$ has no values $>p$, we have $Y^*_{n,m-1} < p$ in both cases. 
\item If the sequence has values $>p$, but a value $\in (q,p)$ precedes all values $>p$, then again $Y^*_{n,m-1} < p$ in both cases. This is because both scans stop at the value $\in (q,p)$ which is not $>p$. 
\item If the sequence has value $> p$, and the first one is \emph{not} preceded by any value $\in (r,p)$, then $Y^*_{n,m-1} > p$ in both cases. This is because both scans stop at a value $> p$ and output it.
\item If the sequence has a value $> p$, but a value $\in (r,q)$ precedes all values $> p$ but \emph{no} value $\in (q,p)$ precedes the first value $> p$: then $Y^*_{n,m-1} >p$ when conditioned on $X_{(1),n} = q$, but $Y^*_{n,m-1} < p$ when conditioned on $X_{(1),n} = r$. This is because the $X_{(1),n} = q$ scan skips over the value $\in (r,q)$, and stops at the value $> p$, whereas the $X_{(1),n} = r$ scan stops at the value $\in (r,q)$.
\end{itemize}

This covers all cases, and proves that for all $r \leq q \leq p$, $\Pr[Y^*_{n,m-1} > p | X_{(1),n} = q] \geq \Pr[Y^*_{n,m-1} > p | X_{(1),n} = r] $.

Finally, we make the same claim for $W_{2,\ell}$: for all $r \leq q$, $\Pr[W_{2,\ell} > p | X_{(1),n} = q] \geq \Pr[W_{2,\ell} > p | X_{(1),n} = r] $. This claim is more straight-forward: $W_{2,\ell}$ is drawn from a uniform distribution on $[X_{(2),\ell},1]$. So clearly, $\Pr[W_{2,\ell} > p | X_{(2),n} = q] \geq \Pr[W_{2,\ell} > p | X_{(2),n} = r] $ whenever $r \leq q$. Moreover, $X_{(2),n}$ is distributed according to the maximum of $n-1$ i.i.d. uniform draws from $[0,X_{(1),n}]$, so the distribution of $X_{(2),n}$ conditioned on $X_{(1),n} = q$ stochastically dominates that of $X_{(2),n}$ conditioned on $X_{(1),n} = r$ whenever $q \geq r$. Both observations together allow us to conclude that for all $r \leq q$, $\Pr[W_{2,\ell} > p | X_{(1),n} = q] \geq \Pr[W_{2,\ell} > p | X_{(1),n} = r] $.

Now we may put all three claims together to prove the lemma. 
\end{proof}

Now with Lemma~\ref{lem:correlation}, we can wrap up the proof. We now know that:

\begin{align*}
\Pr[X_L(n,m) > p] &= \Pr[X_{(1),n} > p] 
   \!\begin{aligned}[t]
&+ \Pr[W_{2,n} > p | X_{(1),n} < p]\cdot \Pr[X_{(1),n} < p]\\
&+ \Pr[Y^*_{n,m-1} > p | X_{(1),n} < p \wedge W_{2,n} < p]\cdot \Pr[X_{(1),n} < p \wedge W_{2,n} < p]
   \end{aligned}\\
&\leq \Pr[X_{(1),n} > p] 
   \!\begin{aligned}[t]
&+ \Pr[W_{2,n} > p | X_{(1),n} < p]\cdot \Pr[X_{(1),n} < p] \\
&+ \Pr[Y^*_{n,m-1} > p | X_{(1),n}<p]\cdot \Pr[X_{(1),n} < p \wedge W_{2,n} < p]
    \end{aligned}\\  
&\leq \Pr[X_{(1),n} > p] 
   \!\begin{aligned}[t]
&+ \Pr[W_{2,n} > p | X_{(1),n} < p] \cdot \Pr[X_{(1),n} < p] \\
&+ \Pr[Z_{(1),n(1+\ln(1+m/n))} > p|X_{(1),n} < p] \cdot  \Pr[X_{(1),n} < p \wedge W_{2,n} < p]
    \end{aligned}\\  
&\leq \Pr[X_{(1),n} > p] + \Pr[W_{2,n} > p \vee Z_{(1),n(1+\ln(1+m/n))}>p | X_{(1),n} < p] \cdot \Pr[X_{(1),n} < p] \\
\end{align*}

At this point, observe that the random variables $W_{2,n}$ and $Z_{(1),n(1+\ln(1+m/n))}$ are independent (and also conditionally independent, conditioned on $X_{(1),n} < p$). Therefore:
\begin{align*}
&\Pr[W_{2,n} < p \wedge Z_{(1),n(1+\ln(1+m/n))}< p | X_{(1),n} < p]\\ 
=& \Pr[W_{2,n} <p | X_{(1),n} < p] \cdot \Pr[Z_{(1),n(1+\ln(1+m/n))}<p | X_{(1),n} < p]
\end{align*}

By Proposition~\ref{prop:key}, we know that $\Pr[W_{2,n} <p | X_{(1),n} < p] \geq \Pr[Z_{(1),n} < p| X_{(1),n} < p]$. Therefore, we get that:
\begin{align*}
&\Pr[W_{2,n} < p \wedge Z_{(1),n(1+\ln(1+m/n))}< p | X_{(1),n} < p]\\ 
\geq &\Pr[Z_{(1),n} < p| X_{(1),n} < p] \cdot \Pr[Z_{(1),n(1+\ln(1+m/n))} < p| X_{(1),n} < p]\\
 =& \Pr[Z_{(1),n} < p] \cdot \Pr[Z_{(1),n(1+\ln(1+m/n))} < p]\\
 =& \Pr[Z_{(1),n(2+\ln(1+m/n))} < p].
\end{align*}
Therefore,
$$\Pr[Z_{(1),n(2+\ln(1+m/n))} > p] \ge \Pr[W_{2,n} > p \vee Z_{(1),n(1+\ln(1+m/n))}>p | X_{(1),n} < p].$$

So substituting all the way back, we get that:
\begin{align*}
[X_L(n,m) > p] &\leq \Pr[X_{(1),n} > p] + \Pr[Z_{(1),n(2+\ln(1+m/n))} > p] \cdot \Pr[X_{(1),n} < p]\\ 
&= \Pr[X_S(n,n(2+\ln(1+m/n))) > p].
\end{align*}

This completes the proof.
\end{proof}

\newpage
\appendix
\section{Lower Bound Examples}\label{sec:lbs}

In this section, we analyze {a new lowerbound}
in detail, and remind the reader of a lower bound from~\cite{FeldmanFR18}. 

\subsection{Theorem~\ref{thm:littlenmain} is tight~\cite{FeldmanFR18}} 

Here, we sketch the construction from~\cite{FeldmanFR18}. There be $n$ bidders and $m > c\cdot n$ items for some absolute constant $c$. Recall that $\erm$ denotes the distribution over values for $m$ items for which each value is drawn i.i.d. from the equal revenue curve {with revenue equal to 1}. Then consider the posted-price mechanism which visits each buyer sequentially (in arbitrary order) and offers her the option to buy any set of $\frac{m}{4n}$ remaining items for price $\pi = \frac{m}{8} \cdot (\ln(m/n)+1)$. Feldman et al. prove the following:

\begin{proposition}[\cite{FeldmanFR18}] The posted-price mechanism described above achieves expected revenue $\Omega(m \cdot n (1+\ln(m/n)))$ for $n$ buyers whose values for the $m$ items are drawn i.i.d.  from $\erm$.
\end{proposition}

On the other hand, the revenue achieved by selling a single item to $n$ buyers whose values are drawn i.i.d. from $\er$ is well-understood. For the sake of completeness we repeat a proof below.

\begin{proposition}[Folklore] $\rev_n(\er)=n$.
\end{proposition}
\begin{proof}
It is clear that $\rev_1(\er) = 1$: for any price $p$, the revenue achieved by setting price $p$ is $p \cdot 1/p = 1$. This immediately implies that $\rev_n(\er) \leq n$: even if the auctioneer had $n$ copies of the item for sale, they would still not get revenue more than $n$.

Moreover, here is an auction that guarantees revenue approaching $n$ as $p \rightarrow \infty$: post price $p$ on the item, and sell to the lexicographically first bidder whose value exceeds $p$. Then the probability of sale is $1 - (1-1/p)^n \geq n/p - \binom{n}{2}/p^2$. So the revenue is at least $n - \binom{n}{2}/p$, which approaches $n$ as $p \rightarrow \infty$. The second-price auction is optimal for $n$ bidders drawn from $\er$, and achieves revenue $n$ (one could separately verify this by directly computing the expected second-highest value of $n$ i.i.d. draws from $\er$, if desired). 
\end{proof}

Together, these propositions claim that $\rev_n(\erm) = \Omega(m \cdot n(\ln(m/n)+1))$, yet also $\vcg_{n+c}(\erm) \allowbreak = m\cdot (n+c)$. Therefore, in order to possibly have $\vcg_{n+c}(\erm) \geq \rev_n(\erm)$ we need to have $c = \Omega(n\ln(m/n))$. We therefore conclude:

\begin{corollary}[\cite{FeldmanFR18}]\label{cor:littletight} The competition complexity of $n$ bidders with additive valuations over $m$ i.i.d., regular items is at least $\Omega(n\ln(m/n))$. 
\end{corollary}

\subsection{Theorem~\ref{thm:bignmain} is tight for the EFFTW benchmark}
In this section, we prove that any analysis starting from the EFFTW benchmark can prove at best a competition complexity of $\sqrt{nm}$. Again, consider the distribution $\erm$. For the subsequent analysis, it will be helpful to instead think of replacing $\er$ by a distribution with CDF $F(x) = 1-1/x$ for $x \in [1,p)$, and $F(p) =1$ (that is, an equal revenue curve truncated at $p$) for $p \rightarrow \infty$. We will not explicitly replace $\er$ with this distribution, and we will always think of $p\rightarrow \infty$ in the subsequent analysis.

Our first lemma states that when considering $\erm$, max of sums in the EFFTW benchmark can be replaced by a sum of maxes. Observe that Lemma~\ref{lem:ERBM} does \emph{not} generally hold for distributions other than $\erm$, and moreover that generally swapping the max and sum results in a benchmark that is unachievable with any finite competition complexity (just consider distributions that are point-masses at $1$). But for equal revenue curves, this equation holds. 

\begin{lemma}\label{lem:ERBM} 
\begin{align*}
&\sum_{j=1}^m \mathbb{E}_{\vec{v} \leftarrow (\erm)^n} \left[\max_{i \in [n]} \left\{ \varphi_j(v_{ij})^+\cdot \mathbb{I}(\vec{v}_i \in R_j) +v_{ij} \cdot \mathbb{I}(\vec{v}_i \notin R_j)\right\}\right]\\
=& \sum_{j=1}^m \mathbb{E}_{\vec{v} \leftarrow (\erm)^n} \left[\max_{i \in [n]} \left\{\varphi_j(v_{ij})^+\cdot \mathbb{I}(\vec{v}_i \in R_j)\right\} +\max_{i \in [n]} \left\{v_{ij} \cdot \mathbb{I}(\vec{v}_i \notin R_j)\right\}\right] \\
 =& nm + \sum_{j=1}^m \mathbb{E}_{\vec{v} \leftarrow (\erm)^n} \left[\max_{i \in [n]} \left\{v_{ij} \cdot \mathbb{I}(\vec{v}_i \notin R_j)\right\}\right] 
\end{align*}
 
\end{lemma}
\begin{proof}
The intuition is roughly as follows: for a single value drawn from the equal revenue curve truncated at $p$, we have $\varphi(v) = 0$ when $v < p$ and $\varphi(p) = p$. So the distribution of virtual values is $0$ with probability $1-1/p$, and $p$ with probability $1/p$. One can therefore informally think of the untruncated equal revenue curve as having virtual value distribution that is $0$ with probability $1$, and $+\infty$ with probability $0$ (yielding an expected value of $1$). Therefore, the argmaximum in the benchmark is some $i$ with $\vec{v}_i \notin R_j$ with probability $1$, and with probability $0$ the argmaximum is some $i$ with $\vec{v}_i \in R_j$ and $\varphi_j(v_{ij}) = +\infty$. As the $\varphi$ term only interferes with probability $0$, we can simply sum the terms instead. Of course, this is quite informal, but provides good intuition for where the proof is going (and why we need to consider $\er$ truncated at $p$ for $p \rightarrow \infty$ to be formal). 

To begin the proof, observe first that whenever $\varphi_j(v_{ij})^+>0$, we have $\varphi_j(v_{ij})^+ = p \geq v_{i'j}$ for all $i'$. That is, whenever a virtual value is non-zero, it is at least as large as any value (because the virtual value is only non-zero when it is equal to $p$, the maximum possible value). Now, let $f^*(\cdot)$ denote the density of $\max_{i \in [n]} \left\{v_{ij} \cdot \mathbb{I}(\vec{v}_i \notin R_j)\right\}$, and let $P^*(x)$ denote the probability that $\max_{i \in [n]} \{\varphi_j(v_{ij})^+\cdot \mathbb{I}(\vec{v}_i\in R_j) \} = 0$, conditioned on $\max_{i \in [n]} \left\{ v_{ij} \cdot \mathbb{I}(\vec{v}_i \notin R_j)\right\} = x$. Then we can write:
\begin{align*}
&\sum_{j=1}^m \mathbb{E}_{\vec{v} \leftarrow (\erm)^n} \left[\max_{i \in [n]} \left\{ \varphi_j(v_{ij})^+\cdot \mathbb{I}(\vec{v}_i \in R_j) +v_{ij} \cdot \mathbb{I}(\vec{v}_i \notin R_j)\right\}\right]\\
 = &m \cdot 
 	\!\begin{aligned}[t] &\left(p \cdot \Pr[\max_{i \in [n]}\{v_{ij}\cdot \mathbb{I}(\vec{v}_i \notin R_j)\} = p] \cdot P^*(p) \right.\\ 
 &+ \int_0^p {x \cdot} f^*(x)\cdot P^*(x) dx\\ 
 &\left. + p\cdot \Pr\left[\max_{i \in [n]} \left\{\varphi_j(v_{ij})^+\cdot \mathbb{I}(\vec{v}_i \in R_j)\right\} > 0\right]\right).
   \end{aligned}\\
 \end{align*}
The left two terms sum the expected contribution from bidders $\notin R_j$, and the last term covers the contribution from bidder $\in R_j$. Let's begin with the left-most term. We claim here that the left-most term approaches $0$ as $p \rightarrow \infty$. To see this, observe that in order to possibly have $\max_{i\in [n]} \left\{v_{ij} \cdot \mathbb{I}(\vec{v}_i \notin R_j)\right\} = p$, there must exist two pairs $(i,j)$ for which $v_{ij} \geq p$. This is because in order to have $\vec{v}_i \notin R_j$, \emph{and} $v_{ij} = p$, we must have $v_{ij'} \geq p$ for some $j' \neq j$. Observe that the probability that this happens is at most $\binom{nm}{2}/p^2$. This yields:

$$\lim_{p \rightarrow \infty} p \cdot \Pr[\max_{i \in [n]}\{v_{ij}\cdot \mathbb{I}(\vec{v}_i \notin R_j)\} = p] \cdot P^*(p) \leq p \cdot \binom{nm}{2}/p^2 = 0.$$

Let's now analyze the integral. Here, we will simply claim that as $p\rightarrow\infty$, $P^*(x) \rightarrow 1$ for all $x$. To see this, let's first understand how conditioning on {$\max_{i \in [n]}\{v_{ij} \cdot \mathbb{I}(\vec{v}_{i} \notin R_j)\} = x$}  
affects the distribution of $\max_{i \in [n]} \left\{\varphi_j(v_{ij})^+\cdot \mathbb{I}(\vec{v}_i \in R_j)\right\}$. 
Observe that conditioning on {$\max_{i \in [n]}\{v_{ij} \cdot \mathbb{I}(\vec{v}_{i} \notin R_j)\} = x$}
 may bias the distribution of the number of indices for which 
 {$\vec{v}_{i} \in R_j$}. For example, if $\max_{i \in [n]} \{v_{ij} \cdot \mathbb{I}(\vec{v}_i \notin R_j)\} = 0$, then we know that $\vec{v}_i \in R_j$ for all $i$. Similarly, if $\max_{i \in [n]} \{v_{ij} \cdot \mathbb{I}(\vec{v}_i \notin R_j)\}$ is large, then it is more likely that $\vec{v}_i \notin R_j$ for many $i$ (because then more terms in the max are non-zero).
 {But this conditioning does not bias the distribution of $\varphi_j(v_{ij})$ for those indices (conditioned on $\vec{v}_i \in R_j$). That is, once we condition in a set $S$ of bidders with $\vec{v}_i \in R_j$ for all $i \in S$, the distribution of $\max_{i \in S} \{\varphi_j(v_{ij})\}$ is independent of $\max_{i \in [n]} \{v_{ij} \cdot \mathbb{I}(\vec{v}_i \notin R_j\}$.} So certainly $1-P^*(x)$ is at most the probability that $n$ independent draws from $\varphi_j(v_{ij})$, conditioned on $\vec{v}_i \in R_j$, are all less than $p$, as for all $x$ there are at most $n$ indicies for which $\vec{v}_i \in R_j$ (simply because there are $n$ bidders). Observe further that the distribution of $\varphi_j(v_{ij})$, conditioned on $\vec{v}_i \in R_j$, is simply the maximum of $m$ i.i.d. draws from $\er$ (truncated at $p$). So the probability that a single one of these draws exceeds $p$ is at most $m/p$ by the union bound. Again taking a union bound over the $n$ draws, the probability that any exceed $p$ is at most $mn/p$. As $p \rightarrow \infty$, this approaches $0$. Therefore, as $p \rightarrow \infty$, $P^*(x) \rightarrow 1$. Observe that we've now shown that:

$$\lim_{p \rightarrow \infty} m \cdot \left(p \cdot \Pr[\max_{i \in [n]}\{v_{ij}\cdot \mathbb{I}(\vec{v}_i \notin R_j)\} = p] \cdot P^*(p) + \int_0^p {x\cdot}  f^*(x)\cdot P^*(x) dx\right) = 0 + m\int_{0}^\infty {x \cdot} f^*(x)dx$$
$$ =  \sum_{j=1}^m \mathbb{E}_{\vec{v} \leftarrow (\erm)^n} \left[\max_{i \in [n]} \left\{v_{ij} \cdot \mathbb{I}(\vec{v}_i \notin R_j)\right\}\right] .$$

We now turn to the final term. First, consider all $nm$ i.i.d. draws from $\er$. There are two ways in which we can have $\max_{i \in [n]} \{\varphi_j(v_{ij})^+\cdot \mathbb{I}(\vec{v}_i \in R_j)\} > 0$. First, maybe the maximum of these $nm$ draws is some bidder $i$'s value for item $j$, and this $v_{ij} = p$. Or, maybe the maximum of these $nm$ draws is some bidder $i$'s value for some other item $j' \neq j$, but there is another bidder whose value for item $j$ exceeds $p$ (implying that there are at least two of the $nm$ draws that exceed $p$). 

For the first case, the probability that the maximum is some bidder $i$'s value for item $j$ is exactly $1/m$. Independently, the probability that this value exceeds $p$ is $1-(1-1/p)^{nm} \in [nm/p-\binom{nm}{2}/p^2, nm/p]$. For the second case, the probability that at least two of the $nm$ draws exceed $p$ is at most $\binom{nm}{2}/p^2$. Therefore, we get that: 
 $$ \Pr\left[\max_{i \in [n]} \left\{\varphi_j(v_{ij})^+\cdot \mathbb{I}(\vec{v}_i \in R_j)\right\} > 0\right] \in  [n/p -\binom{nm}{2}/(mp^2), n/p+\binom{nm}{2}/p^2].$$

And therefore $$ p\cdot \Pr\left[\max_{i \in [n]} \left\{\varphi_j(v_{ij})^+\cdot \mathbb{I}(\vec{v}_i \in R_j)\right\} > 0\right] \in  [n-\binom{nm}{2}/(mp), n+\binom{nm}{2}/p].$$

Implying that:

$$ \lim_{p \rightarrow \infty}p \cdot \Pr\left[\max_{i \in [n]} \left\{\varphi_j(v_{ij})^+\cdot \mathbb{I}(\vec{v}_i \in R_j)\right\} > 0\right] =n.$$

This completes the analysis of the right-most term, as now:

$$\lim_{p \rightarrow \infty} m\cdot p\cdot \Pr\left[\max_{i \in [n]} \left\{\varphi_j(v_{ij})^+\cdot \mathbb{I}(\vec{v}_i \in R_j)\right\} > 0\right] = \mathbb{E}_{\vec{v} \leftarrow (\erm)^n} \left[\max_{i \in [n]} \left\{\varphi_j(v_{ij})^+\cdot \mathbb{I}(\vec{v}_i \in R_j)\right\}\right] = nm.$$

Putting both parts together proves the lemma.

\end{proof}

Now, we just need to analyze $\sum_{j=1}^m \mathbb{E}_{\vec{v} \leftarrow (\erm)^n} \left[\max_{i \in [n]} \left\{v_{ij} \cdot \mathbb{I}(\vec{v}_i \notin R_j)\right\}\right] $.

\begin{proposition}\label{prop:ERvals} Let $n \geq 4m$. Then: $$\sum_{j=1}^m \mathbb{E}_{\vec{v} \leftarrow (\erm)^n} \left[\max_{i \in [n]} \left\{v_{ij} \cdot \mathbb{I}(\vec{v}_i \notin R_j)\right\}\right]  = \geq \frac{m\sqrt{mn}}{{14}}.$$
\end{proposition}
\begin{proof}

We begin by considering the probability that the $\ell^{th}$ highest of $nm$ draws from $\er$ (or any distribution) is some bidder $i$'s value for item $j$, \emph{and} $\vec{v}_i \notin R_j$ \emph{and} for all $\ell' < \ell$, either it is some bidder $i'$'s value for an item $j' \neq j$, or $v_{i'} \in R_j$. That is, we are interested in computing the probability that the $\ell^{th}$ highest of the $nm$ i.i.d. draws is the maximum value (times $ \mathbb{I}(\vec{v}_i \notin R_j)$) for item $j$. Observe that this probability is well-defined: it does not depend on the particular values for the $nm$ draws (and does not even depend on the distribution from which they are drawn). Indeed, whether the desired event occurs or not is only a function of how the draws are permuted among the $nm$ values of bidders for items. 

So denote by $P_\ell$ the probability that the $\ell^{th}$ highest of $nm$ draws is {some} bidder $i$'s value for item $j$ \emph{and} $\vec{v}_i \notin R_j$ \emph{and} for all $\ell' < \ell$, either it is some bidder $i'$'s value for an item $j' \neq j$, or $\vec{v}_{i'} \in R_j$. {Let $E_1$ denote the event that the $\ell^{th}$ highest is some bidder's value for item $j$. Let $E_2$ denote the probability that $\vec{v}_{i(\ell)} \notin R_{j(\ell)}$, where the $\ell^{th}$ highest is assigned to bidder $i(\ell)$ and item $j(\ell)$. Let $E_3$ denote the event that for all $\ell' < \ell$, either $j(\ell') \neq j(\ell)$ or $\vec{v}_{i(\ell')}\in R_{j(\ell)}$.} Then it is easy to see that $\Pr[E_1] = 1/m$. Moreover, events $E_2$ and $E_3$ only involve where the highest $\ell-1$ draws go. In particular, observe that both $E_2$ and $E_3$ are independent of $E_1$. Moreover, observe that $E_3$ is only more likely to occur conditioned on $E_1$ and $E_2$: if one of the top $\ell-1$ draws is for the same bidder as the $\ell^{th}$, it is necessarily \emph{not} for the same item $j$, and also necessarily \emph{not} helping to put some other $\vec{v}_{i'} \notin R_j$ (by being larger than $v_{i'j})$ (hence making it more likely that all of the top $\ell-1$ draws are permuted to some item other than $j$, and also that those values which are still permuted to item $j$ are $\in R_j$). Therefore:

\begin{align*}
P_\ell &= \Pr[E_1] \cdot \Pr[E_2 | E_1] \cdot \Pr[E_3 |E_2 \wedge E_1]\\
 &= {\Pr[E_1]} \cdot \Pr[E_2] \cdot \Pr[E_3 | E_2 \wedge E_1] \\
&\geq {\Pr[E_1]} \cdot \Pr[E_2] \cdot \Pr[E_3]\\
&=\Pr[E_3 \wedge E_1] \cdot \Pr[E_2].
\end{align*}

Let's first analyze the probability of $E_2$. Observe that event $E_2$ does \emph{not} happen only if each of the top $\ell-1$ draws are permuted to a different bidder than $\ell$. So:

$$\Pr[E_2]  = 1-\left(1-\frac{m-1}{mn-1}\right)\cdot\left(1- \frac{m-1}{mn-2} \right)\cdot \dots \cdot \left(1-\frac{m-1}{mn-\ell+1} \right)\in [1-(1-\frac{m-1}{mn-1})^{\ell-1}, (\ell-1)/n].$$

To see the first equality, observe that $\frac{mn-1}{mn-\ell'}$ is the probability that the $\ell'^{th}$-highest draw is permuted to a different bidder than the $\ell^{th}$, conditioned on the first thru $(\ell'-1)^{th}$ draws all being permuted to a different bidder than the $\ell^{th}$. This is because there are $m-1$ items left for the same bidder, and $mn-\ell'$ remaining (bidder, item) pairs in total. To see the upper bound, observe that the probability that a single $\ell'<\ell$ is permuted to the same bidder as $\ell$ is exactly $(m-1)/(mn-1)$. So taking a union bound over all $\ell' < \ell$, we get an upper bound on the probability that \emph{any} $\ell'$ is permuted to the same bidder as $\ell$ is at most $(\ell-1)(m-1)/(mn-1) \leq (\ell-1)m/(mn) = (\ell-1)/n$. The lower bound follows by just observing that there are $\ell-1$ terms in the product, and each term is at {most} $(1-\frac{m-1}{mn-1})$. 

Finally, we prove one minor technical lemma to argue that when $\ell-1 \leq n/2$, $1-(1-\frac{m-1}{mn-1})^{\ell-1} \geq (1-\ln(2))\cdot (\ell-1)/{(2n)}$.

\begin{lemma}\label{lem:moremath} Let $\ell -1\leq n/2$ and $n \geq 2$. Then $1-(1-\frac{m-1}{mn-1})^{\ell-1} \geq (1-\ln(2)) \cdot (\ell-1)/{(2n)}$.
\end{lemma}
\begin{proof}
We start by searching for a constant $c$ such that for all $x \in [0,1/2]$, and all $y$ such that $xy \in[0,1/2]$, we have $1-(1-x)^y \geq cxy$. We will then use this constant to prove the lemma statement. Let's first fix $x$ and $c$, and minimize $1-(1-x)^y - cxy$ over all $y \in [1,1/(2x)]$ (this search is well-defined unless $x = 0$, in which case $1-(1-x)^y - cxy = 0$ for all $y$). The derivative with respect to $y$ is $-\ln(1-x) (1-x)^y - cx$, and the second derivative is $\ln(1-x)^2 (1-x)^y \geq 0$. So if the first derivative is positive at $y=1$, it is positive on the entire interval $[1,1/(2x)]$. So now we wish to see how small the first derivative can be, as a function of $x$, when $y = 1$. 

This is again single-variate optimization: minimize $-\ln(1-x) \cdot (1-x) - cx$ on $[0,1/2]$. The derivative is $1 +\ln(1-x) -c$. Observe that if $c \leq 1-\ln(2)$, then the derivative is $\geq 0$ on the entire range $(0,1/2)$. This means that for $c \leq 1-\ln(2)$, the minimizer occurs at $x= 0$, which is $-\ln(1)\cdot 1 - c\cdot 0 = 0$. So at this point, we conclude that when $c \leq 1-\ln(2)$, the first derivative with respect to $y$ is non-negative at $y=1$, and therefore positive on the entire interval $[1,1/(2x))$. This means that the minimum occurs at $y = 1$ (and we will restrict ourselves from now on to $c \leq 1-\ln(2)$). 

Now that we know the minimum occurs at $y=1$, our remaining minimization is trivial: minimize $1-(1-x)-cx = (1-c)x$, which is achieved at $x = 0$, and indeed $c\cdot 0 \geq 0$. So we have proven that when $c \leq 1-\ln(2)$, the minimum value for $1-(1-x)^y - cxy$ over all $x \in [0,1/2]$ and $y \in [1,1/(2x))$ occurs at $x = 0$, $y= 1$, and the value is $\geq 0$. This proves that $1-(1-x)^y  \geq (1-\ln(2))xy$ whenever $x \in [0,1/2]$ and $xy \in [0,1/2]$. 

Now we wish to apply the above fact when $x = \frac{m-1}{mn-1}$ and $y = \ell-1 \leq n/2$. Indeed, observe first that when $n \geq 2$, we have that $x \leq 1/2$. Moreover, when $y \leq n/2$ we have $x\cdot y = \frac{nm-n}{2mn-2} \leq 1/2$. {Applying the previous work gives $1-(1-\frac{m-1}{mn-1})^{\ell-1} \geq (1-\ln(2)) \cdot (\ell-1)\cdot \frac{m-1}{mn-1}$. Note that for $m \ge 2$, $\frac{m-1}{mn-1} \ge 1/(2n)$. Therefore we can conclude the lemma statement.}
\end{proof}

Now with Lemma~\ref{lem:moremath}, we can conclude that when $\ell-1 \leq n/2$, we have:

$$\Pr[E_2] \in [(1-\ln(2))(\ell-1)/{(2n)}, (\ell-1)/n].$$

{Now, let's analyze the probability of $E_3 \wedge E_1$. Observe that first, $\Pr[E_1] = 1/m$, and $E_3$ and $E_1$ are independent. So let's look at $\Pr[E_3 | E_1]$. Observe that $P_{\ell'}$ is \emph{exactly} the probability that $E_3$ does \emph{not} occur, conditioned on $E_1$, because the $\ell'^{th}$ highest value has $j(\ell') = j(\ell)$ and $\vec{v}_{i(\ell')} \notin R_{j(\ell')}$. If none of these events occur, then certainly event $E_3$ occurs. Therefore:}

$$\Pr[E_3 | E_1] \geq 1- \sum_{\ell' < \ell} P_{\ell'}, \Pr[E_3 \wedge E_1] \geq (1-\sum_{\ell' < \ell} P_{\ell'}) /m.$$

And therefore we can conclude that:
\begin{align}\label{eq:final}
P_\ell &\geq \Pr[E_1\wedge E_3] \cdot \Pr[E_2]
 \geq 1/m \cdot \left(1-\sum_{\ell'<\ell} P_{\ell'}\right) \cdot \left(1-\ln(2)\right)(\ell-1)/{(2n)}.
\end{align}

{Observe also that $P_\ell \leq \Pr[E_1 \wedge E_2] = \Pr[E_1] \cdot \Pr[E_2] \leq \frac{\ell-1}{mn}$.} In particular, we can use this \emph{upper bound} on $P_{\ell'}$, $\ell' < \ell$ to derive a cleaner \emph{lower bound} on $P_\ell$. Indeed, for any $\ell < \sqrt{mn} \leq n/2$ (the last inequality is because we assume that $n \geq 4m$. This is the only place we use this assumption, but it is a key step), we have that:

$$\sum_{\ell' < \ell} P_{\ell'} \leq \sum_{\ell' < \sqrt{mn}} (\ell'-1)/mn \leq \frac{\sqrt{mn}^2}{2mn} = 1/2.$$

{By plugging the above in inequality~\ref{eq:final}} we can conclude that for all $\ell < \sqrt{mn}$, we have:

$$P_\ell \geq \frac{(1-\ln(2))\cdot (\ell-1)}{{4}nm}.$$

To recap: we have now shown that, independent of the particular values drawn, the $\ell^{th}$ highest of $nm$ draws is equal to $\max_{i \in [n]}\{v_{ij} \cdot \mathbb{I}(\vec{v}_i \in R_j)\}$ with probability at least $P_\ell \geq \frac{(1-\ln(2))\cdot (\ell-1)}{{4}mn}$,
 for all $\ell < \sqrt{mn}$. So the only remaining step is to compute the expected value of the $\ell^{th}$ highest of $nm$ draws from $\er$. 

\begin{lemma}\label{lem:ERellth} Let $V_{x, y}$ denote the $x^{th}$ highest of $y$ i.i.d. draws from $\er$. Then $\mathbb{E}[V_{x,y}] = y/(x-1)$.\footnote{Observe that when $x = 1$, the expected value of the highest of $\geq 1$ draws from $\er$ is indeed $+\infty$.}
\end{lemma}
\begin{proof}
One approach would be to explicitly write out the integral for the $x^{th}$ highest of $y$ draws. This is tedious. Instead, we will count the revenue of an auction in two different ways. Consider an $(x-1)$-unit auction with $x-1$ copies of the same item. There are $y$ unit-demand bidders with values drawn i.i.d. from $\er$. Then as all virtual values are non-negative and there is no ironing, the revenue-optimal auction is to set the $x^{th}$ highest bid as the price and let the $(x-1)$ highest bidders get the item. The expected revenue of this auction is exactly $V_{x,y} \cdot (x-1)$. 

On the other hand, we claim that the optimal revenue for this setting is exactly $y$. To see this, observe that clearly the optimal revenue is at most $y$, as even with $y$ copies of the item, we could not get revenue more than $1$ per bidder. On the other hand, consider the mechanism that posts a price $p$ and lets any bidder willing to pay $p$ get the item, for $p \rightarrow \infty$. Then the probability that at least one bidder chooses to pay is at least $y/p - \binom{y}{2}/p^2$. So the revenue is at least $y - \binom{y}{2}/p$, which approaches $y$ as $p \rightarrow \infty$. So the optimal revenue for this setting is indeed $y$. 

Therefore, we get that $y = V_{x,y} \cdot (x-1)$ and $V_{x,y} = \frac{y}{x-1}$. 
\end{proof}

Now we can put everything together. We have argued that $\max_{i \in [n]} \{v_{ij} \cdot \mathbb{I}(\vec{v}_i \notin R_j)\}$ is equal to the $\ell^{th}$ highest of $nm$ draws with probability 
$P_\ell \geq \frac{(1-\ln(2))(\ell-1)}{{4}mn}$
 for all $\ell < \sqrt{mn}$, and also argued that the expected value of the $\ell^{th}$ highest of $mn$ draws is $\frac{mn}{\ell-1}$. Therefore, for all $j$ we get:

$$\mathbb{E}_{\vec{v} \leftarrow (\erm)^n} \left[\max_{i \in [n]}\left\{ v_{ij} \cdot \mathbb{I}(\vec{v}_i \notin R_j)\right\}\right] \geq \sum_{\ell<\sqrt{nm}} P_\ell \cdot \mathbb{E}[V_{\ell,nm}] \geq {\sum_{\ell<\sqrt{nm}}} \frac{(1-\ln(2))(\ell-1)}{{4}nm} \cdot \frac{nm}{\ell-1} \geq \sqrt{nm}/{14}. $$

Summing over all $j$ yields the proposition statement.
\end{proof}

Now with Lemma~\ref{lem:ERBM} and Proposition~\ref{prop:ERvals}, we can conclude the following:

$$\sum_{j=1}^m \mathbb{E}_{\vec{v} \leftarrow (\erm)^n} \left[\max_{i \in [n]} \left\{ \varphi_j(v_{ij})^+\cdot \mathbb{I}(\vec{v}_i \in R_j) +v_{ij} \cdot \mathbb{I}(\vec{v}_i \notin R_j)\right\}\right] \geq nm+m\sqrt{nm}/{14}.$$

This immediately implies the following corollary:

\begin{corollary}\label{cor:bigntight} If one compares to the EFFTW benchmark (which upper bounds the expected revenue) instead of the expected revenue, then the competition complexity of $n$ bidders with additive valuations over $m\leq n/4$ i.i.d., regular items is at least $\sqrt{nm}/{14}$. 

\end{corollary}
\begin{proof}
Simply recall that $\rev_{n+c}(\er) = n+c$, and $\vcg_{n+c}(\erm) = m\cdot \rev_{n+c}(\er) = m(n+c)$. So in order to have $\vcg_{n+c}(\erm) \geq nm+\sqrt{nm}/{14}$, we must have $c \geq \sqrt{nm}/{14}$. 
\end{proof}

\section{Competition complexity is not independent of $n$.}
In this section, we show that while indeed the ``true'' competition complexity in the big-$n$ case may be better than what is achievable by comparing to the EFFTW benchmark, it is not independent of $n$. Specifically, we will argue that for all $n$ and $m= 2$ items, the optimal mechanism for $(\erm)^n$ achieves revenue at least $2n+\ln(n)/10$. 

\begin{proposition}\label{prop:erlogs} $\rev_n(\er^2) \geq 2n+\ln(n)/10$.

\end{proposition}
\begin{proof}
Consider the following mechanism:
\begin{itemize}
\item Bidders select whether to be ``high'', ``low'', or ``medium''.
\item First, all high bidders are processed in random order. When processed, a high bidder can get both remaining items for price $p$.
\item Next, all medium bidders are processed in random order. When processed, a medium bidder will get a random remaining item for price $q$.
\item Low bidders get no items.
\end{itemize}

We will think of $p \rightarrow \infty$, and take $q:=\sqrt{n}$. First, it is easy to see that any bidder who values the grand bundle (of both items) at least as $2q$ will choose to be a high or medium bidder. Let's now count the expected number of such bidders. 

\begin{lemma}$$\Pr_{(v_1,v_2) \leftarrow \er^2}\left[v_1+v_2 \geq 2q\right] = \frac{2}{2q-1} +\frac{(q-\frac{1}{2})\cdot \ln(2q-1) - q}{q^2(2q-1)}$$
\end{lemma}

\begin{proof}
There are two ways that that we can have $v_1+v_2 \geq 2q$. First, maybe $v_1 \geq 2q-1$. In this case, as $v_2 \geq 1$, surely $v_1+v_2 \geq 2q$. Second, maybe $v_1 < 2q-1$ and $v_2 \geq 2q-v_1$. So the probability of both cases together is: $\frac{1}{2q-1} + \int_1^{2q-1} \frac{1}{x^2} \cdot \frac{1}{2q-x} dx$. Also:
\begin{align*}
&\int_1^{2q-1} \frac{1}{x^2} \cdot \frac{1}{2q-x} dx\\
=&\frac{1}{4} \cdot \left[\frac{x\ln(2q-x) + 2q - x\ln x}{q^2x}\right]_1^{2q-1}\\
=& -\frac{1}{4}\cdot \left(\frac{2q - (2q-1)\cdot \ln(2q-1)}{q^2(2q-1)} - \frac{\ln(2q-1) + 2q}{q^2}\right)\\
=&\frac{1}{2} \cdot \frac{-q+q(2q-1)+(2q-1)\cdot \ln(2q-1)}{q^2(2q-1)}\\
=& \frac{q^2 - q +(q-\frac{1}{2})\ln(2q-1)}{q^2(2q-1)}.
\end{align*}

Adding back in the $\frac{1}{2q-1}$, we get the lemma statement.
\end{proof}

We will want to take $q$ small enough so that in expectation $n^{\Omega(1)}$ bidders wish to be medium{/high}. This will let us claim that he expected number of medium/high bidders concentrates around its expectation. But we also want $q$ to be big enough so that we can get some revenue from these cases. 

\begin{lemma}Consider any set of $n-1$ bidders, and let $100 \leq q \leq \sqrt{n}$ and $p \rightarrow \infty$. Then with probability at least $1-e^{-n\ln^2(q)/(128q^2)} = 1-o(1/n)$, none of these bidders are high, and at least $\frac{n}{q} + \frac{n\ln(q)}{8q^2}$ are medium.
\end{lemma}
\begin{proof}
The number of bidders with $v_1+v_2$ exceeding {$2q$}
 is a sum of independent $\{0,1\}$ random variables whose expected value exceeds $\frac{n}{q} + \frac{n\ln(q)}{4q^2}$. (We have assumed $q \geq 100$ in order to simplify the second term). 

We therefore wish to understand the probability that the actual number of medium/high bidders is at least its expectation minus $\frac{n\ln(q)}{8q^2}$. This is a simple application of the Chernoff bound with $\mu \geq \frac{n}{q}$ and $\delta \leq \frac{\ln(q)}{8q}$. So the probability of this deviation is at most $e^{-n\ln^2(q)/(128q^2)}$. When $q < \sqrt{n}$, this probability is $o(1/n)$. 

Now we just need to recall that this is the probability that at least $\frac{n}{q} + \frac{n\ln(q)}{8q^2}$ bidders are high or medium. Observe further that the probability that any are high bidders is at most $4n/p$ (some bidder must value some item at at least $p/2$, which occurs with probability at most $4n/p$ by union bound), which approaches $0$ as $p \rightarrow \infty$. 
\end{proof}

Now, let's see what value a bidder would need to have in order to prefer to be high instead of medium.

\begin{lemma}Assuming that all other bidders are medium or low, and there are $k-1$ medium bidders, a bidder with $v_1+v_2 \geq \frac{pk-2q}{k-1}$ chooses to be a high bidder instead of medium.
\end{lemma}
\begin{proof}
If the bidder chooses to be medium, then they will get each item with probability $1/k$, and pay $2q/k$ in expectation. If they choose to be high, they will get each item with probability $1$, and pay $p$. So in order to prefer this, they would need to have $(v_1+v_2)\cdot (1-1/k) \geq p-2q/k$. Rearranging yields the lemma.
\end{proof}

Now we analyze the revenue from two possible cases. Below for notational simplicity, define $k:= \frac{n}{q}+\frac{n\ln(q)}{8q^2}$.

\begin{itemize}
\item \textbf{Case One: Some bidder has $v_1+v_2 \geq pk/(k-1)$.} In this case, with probability at least $1-o(1/n)$, there are no high bidders and at least $k$ medium bidders among the other $n-1$. Conditioned on this, the bidder will choose to be high, and we get revenue $p$. The probability that this case occurs is at least $2n(1-1/k)/p - \binom{2n}{2}(1-1/k)^2/p^2$ (as it occurs whenever anyone values some item above $p$). Conditioned on this, with further probability $1-o(1/n)$, this buyer chooses to be high and we get revenue $p$. So the total revenue from these cases is at least: $(1-o(1/n))\cdot (p \cdot (2n(1-1/k)/p - \binom{2n}{2}(1-1/k)^2/p^2)) =2(1-1/k)n-o(1)$ as $p \rightarrow \infty$. 
\item \textbf{Case Two: no bidder has $v_1+v_2 \geq pk/(k-1)$.} In this case, we know that with probability $1-o(1/n)$, there are at least two (in fact, many more) bidders who are medium. So we get revenue $2q$ in these cases. Observe that as $p\rightarrow \infty$, we are in this case with probability approaching $1$, so the total revenue from these cases is $2q \cdot (1-o(1/n)) = 2q - o(1)$ as $p \rightarrow \infty$. 
\end{itemize}

So our revenue is $2n(1-1/k) + 2q - o(1)$. Recall that $k:= \frac{n}{q} + \frac{n\ln(q)}{8q^2}=\frac{8nq+n\ln(q)}{8q^2}$. Therefore our revenue is:
\begin{align*}
&2n - \frac{16nq^2}{8nq+n\ln(q)} + \frac{16nq^2+2nq\ln(q)}{8nq+n\ln(q)}\\
=& 2n +\frac{2q\ln(q)}{8q+\ln(q)} \geq 2n+\ln(q)/5.
\end{align*}

Recall that we required {$q \leq \sqrt{n}$}
 in order for our calculations to be valid, so by setting $q:=\sqrt{n}$, our revenue is at least $2n+\ln(n)/10$. 

\end{proof}

And now we can conclude that the competition complexity cannot be independent of $n$:

\begin{corollary}\label{cor:logbign} The competition complexity of $n$ bidders with additive valuations over $m=2$ i.i.d., regular items is at least $\ln(n)/10$. In particular, there exists no function $f(\cdot)$ only of $m$ such that the competition complexity of $n$ bidders with additive valuations over $m$ i.i.d. items is $f(m)$. 

\end{corollary}
\section{Omitted Proofs from Section~\ref{sec:prelim}}\label{app:prelim}

\begin{proof}[Proof of Fact~\ref{fact:one}]  The proof follows quickly from Myerson's lemma, stated below:
\begin{lemma}[Myerson's Lemma~\cite{Myerson81}]\label{lem:myerson} Consider any Bayesian Incentive Compatible mechanism for a single item with payment rule $P(\cdot)$ and allocation rule $X(\cdot)$. That is, on bids $\vec{v}$, the mechanism charges bidder $i$ $P_i(\vec{v})$ and awards bidder $i$ the item with probability $X_i(\vec{v})$. Then for all $i$, the expected payment made by bidder $i$ is equal to bidder $i$'s expected virtual welfare. That is:
$$\mathbb{E}_{\vec{v} \leftarrow D^n}\left[P_i(\vec{v})\right] = \mathbb{E}_{\vec{v}\leftarrow D^n}\left[X_i(\vec{v}) \cdot \varphi_D(v_i)\right],$$
$$\mathbb{E}_{\vec{v} \leftarrow D^n}\left[P_i(\vec{v})\right] \leq \mathbb{E}_{\vec{v}\leftarrow D^n}\left[X_i(\vec{v}) \cdot \overline{\varphi}_D(v_i)\right].$$
\end{lemma}

Now, consider the single-bidder auction that simply sets a price of $v$. Then it's expected revenue is simply $v \cdot \Pr_{w \leftarrow D}[w \geq v]$. The expected virtual welfare is $\mathbb{E}_{w \leftarrow D}[\varphi_D(w) \cdot \mathbb{I}(w \geq v)]$. Therefore, by Myerson's Lemma:

$$v =\frac{\mathbb{E}_{w \leftarrow D}\left[\varphi_D(w) \cdot \mathbb{I}(w \geq v)\right]}{\Pr_{w \leftarrow D}\left[w \geq v\right]} = \mathbb{E}_{w \leftarrow D_{\geq v}}\left[\varphi_D(w)\right].$$

The proof for ironed virtual values follows identically, after replacing the left-most equality with inequality.
\end{proof}

\section{Omitted Proofs from Section~\ref{sec:multi}}\label{app:multi}
\subsection{Little $n$ proofs}
\begin{proof}[Proof of Observation~\ref{obs:EFFTW}]
For all $\vec{v}$, if $\vec{v}_{(1)} \notin R_j$, then both random variables inside the expectation take value $v_{(1)j}$. If $\vec{v}_{(1)}\in R_j$, then the random variable for the left-hand expectation is at most $\max\{\overline{\varphi}_j(v_{(1)j}), v_{(2)j}\}$ (note that the $^+$ is no longer necessary as we're taking a maximum with $v_{(2)j} \geq 0$ anyway). The right-hand expectation is exactly this. Because the right-hand random variable is larger for all $\vec{v}$, the expectation is larger as well.
\end{proof}

\begin{proof}[Proof of Proposition~\ref{prop:littlenmain}] We again begin with a coupling. First, couple the draws $X_{i,n}:=F_j(v_{ij})$ for all $i$. Next, couple the quantiles $Y_{\ell, m-1} = F_\ell(v_{(1)\ell})$ for $ \ell < m, \ell \neq j$, and $Y_{j, m-1}=F_m(v_{(1)m})$ (if $j \neq m$, otherwise there is no $Y_{m, m-1}$ to define). Crucially, observe that this is a valid coupling, as values for items are drawn independently (in particular, conditioned on drawing $v_{ij}$ for all $j$, the quantiles $F_\ell(v_{i\ell})$ are still i.i.d. and uniform from $[0,1]$). Observe now the following:
\begin{itemize}
\item Always, $v_{(\ell)j} = F_j^{-1}(X_{(\ell),n})$ for all $\ell$.
\item Whenever $\vec{v}_{(1)} \in R_j$, $X'_L(n,m)=X_{(1),n}$. Therefore:
$$\mathbb{E}_{\vec{v} \leftarrow D^n}\left[\max\{v_{(1)j}\cdot \mathbb{I}(\vec{v}_j \notin R_j), \overline{\varphi}_j(v_{(1)j}), v_{(2)j}\}\cdot \mathbb{I}(\vec{v}_{(1)} \in R_j)\right] $$
$$= \mathbb{E}\left[\max\{\overline{\varphi}_j(F_j^{-1}(X'_L(n,m))),F_j^{-1}(X_{(2),n}) \} \cdot \mathbb{I}(X'_L(n,m) = X_{(1),n})\right].$$
\item Conditioned on $\vec{v}_{(1)} \notin R_j$, $X'_L(n,m)$ is a uniformly random sample from $[X_{(1),n},1]$. This is because there is some strictly positive number of $\ell$ such that $Y_{\ell, m-1} > X_{(1),n}$. Conditioned on being $>X_{(1),n}$, each such value is drawn uniformly from $[X_{(1),n},1]$. And then $X'_L(n,m)$ picks one of them uniformly at random. Using Fact~\ref{fact:one}, we get:
$$\mathbb{E}_{\vec{v} \leftarrow D}\left[v_{(1)j} \cdot \mathbb{I}(\vec{v}_{(1)} \notin R_j)\right] \leq \mathbb{E}\left[\overline{\varphi}_j(F_j^{-1}(X'_L(n,m)))\cdot \mathbb{I}(X'_L(n,m) \neq X_{(1),n})\right].$$
\end{itemize}
Summing the left-hand side of both equations is exactly $\mathbb{E}_{\vec{v} \leftarrow D^n} \left[\max \left\{v_{(1)j}\cdot \mathbb{I}(\vec{v}_{(1)} \notin R_j), \overline{\varphi}_j(v_{(1)j}), v_{(2)j} \right\} \right]$. Summing the right-hand side is clearly upper bounded by $\mathbb{E}\left[\max\left\{\overline{\varphi}_j(F_j^{-1}(X'_L(n,m))),F_j^{-1}(X_{(2),n})\right\}\right]$.

\end{proof}

 For subsequent proofs, we will need one basic fact about maximums of random variables:

\begin{fact}\label{fact:basic} For any three random variables $X, Y, Y'$ such that for all $x$, $\mathbb{E}[Y' | Y = y, X = x] \geq y$, $\mathbb{E}[\max\{X,Y'\}] \geq \mathbb{E}[\max\{X, Y\}]$. 
\end{fact}
\begin{proof}
In fact, we show that for all $y, x$, $\mathbb{E}[\max\{X, Y'\} | Y = y,  X = x] \geq \max\{x, y\}$, which implies the desired statement. There are two cases to consider. First, perhaps $x \geq y$. In this case, we clearly have $\max \{x, Y'\} \geq x$ with probability $1$. Therefore, $\mathbb{E}[\max\{X , Y'\} | Y = y, X= x] \geq x$ as well. Second, perhaps $y > x$. In this case, by hypothesis we have $\mathbb{E}[Y' | Y = y, X = x] \geq y$. So clearly $\mathbb{E}[\max \{X, Y'\} | Y = y, X = x] \geq y$. This covers both cases and proves the fact.
\end{proof}

\begin{proof}[Proof of Corollary~\ref{cor:littlenmain}]
We first consider the three random variables $X = \overline{\varphi}_j(F_j^{-1}(X'_L(n,m))), Y = F_j^{-1}(X_{(2),n}), Y' =  \overline{\varphi}_j(F_j^{-1}(W_{2,n}))$. Then indeed, conditioned on $X = x$ and $Y = y$, $W_{2,n}$ is a uniformly random draw from $[X_{(2),n},1]$. Fact~\ref{fact:one} therefore concludes that $\mathbb{E}[Y' | X = x, Y = y] = y$, allowing an application of Fact~\ref{fact:basic} to conclude $\mathbb{E}\left[\max\left\{\overline{\varphi}_j(F_j^{-1}(X'_L(n,m))),F_j^{-1}(X_{(2),n})\right\}\right] \leq \mathbb{E}\left[\max \left\{\overline{\varphi}_j(F_j^{-1}(X'_L(n,m))),\overline{\varphi}_j(F_j^{-1}(W_{2,n}))\right\}\right]$. The corollary follows by observing that\\ $\max \left\{\overline{\varphi}_j(F_j^{-1}(X'_L(n,m))),\overline{\varphi}_j(F_j^{-1}(W_{2,n}))\right\} = \overline{\varphi}_j(\max\{X'_L(n,m),W_{2,n}\}) = \overline{\varphi}_j(X_L(n,m))$. 
\end{proof}

\begin{proof}[Proof of Corollary~\ref{cor:single}]
The proof is nearly-identical to that of Corollary~\ref{cor:reducesingle} and omitted (the only difference is that we save a ``$+1$'' in the regular case because our current benchmark no longer has the ``$^+$'' on the virtual values).
\end{proof}

\subsection{Big $n$ proofs}

\begin{proof}[Proof of Proposition~\ref{prop:stepone}]
Consider drawing quantiles $q_1,\ldots, q_n$, and also $\{q'_{i,{k}}\}_{i \in [\ell-1], {k} \in [m-1]}$. 
Couple draws from $\vec{v}$ from $D^n$ and $\vec{w}$ from $D_j^{n+(\ell-1)(m-1)}$ as follows:
\begin{itemize}
\item Set $v_{ij} = w_{ij} = F_j^{-1}(q_i)$ for all $i$. Relabel the indices so that $v_{(1)j} \geq v_{(2)j} \geq \ldots \geq v_{(n)j}$,  $w_{(1)j} \geq w_{(2)j} \geq \ldots \geq w_{(n)j}$, and $q_{(1)} \geq q_{(2)} \geq \ldots \geq q_{(n)}$.
\item For $i \in [\ell-1]$, set $v_{(i)k} = F_k^{-1}(q'_{i,k})$ for all $k \neq j$. Set $v_{(i)m} = F_m^{-1}(q'_{i,j})$ (unless $j = m$, in which case $v_{(i)m}$ is already set, and there is no $q'_{i,m}$). 
\item For $i \in [\ell-1]$, set $w_{(i)k} = F_j^{-1}(q'_{i,k})$ {for all $k \neq j$. Set $w_{(i)m} = F_j^{-1}(q'_{i,j})$ (unless $j = m$, in which case $w_{(i)m}$ is already set, and there is no $q'_{i,m}$).} As all $w$s are drawn from $D_j$, interpret all $n+(\ell-1)(m-1)$ such draws as values of a single bidder for item $j$. Let $w_{(\ell)}$ denote the $\ell^{th}$ largest of these draws.
\item Observe that $w_{(\ell)} \geq w_{(\ell)j} = v_{(\ell)j}$. 
\end{itemize}

Now, we consider the random variables inside the left-hand and right-hand expectations. Let now $i^*$ denote the minimum $i$ such that $\vec{v}_{(i)} \notin R_j$. Observe that when $i^* < \ell$, this is also the minimum $i$ such that there exists a $k \neq j$ with $F_j(w_{{(i)}j}) < F_j(w_{{(i)}k})$.\footnote{If $D_j$ has no point masses, we could instead just write $w_{{(i)}j}< w_{{(i)}k}$, but we write it like this to be careful in the case of point masses.} Therefore, $i^* \geq \ell$ only if no $i,k$ exists for which $F_j(w_{{(i)}j}) < F_j(w_{{(i)}k})$. So first, consider the possibility that $i^* \geq \ell$. In this case, the contribution to the benchmark is upper bounded by $\max\{\overline{\varphi}_j(v_{(1)j}), v_{(\ell)j}\}$. But we have that $\overline{\varphi}_j(w_{(1)}) \geq \overline{\varphi}_j(v_{(1)j})$, and also $w_{(\ell)} \geq v_{(\ell){j}}$.
Therefore, we can conclude that:

$$\mathbb{E}_{\vec{v} \leftarrow D^n} \left[\left(\max_{i \in [n]}\left\{\overline{\varphi}_j(v_{ij})^+ \cdot \mathbb{I}(\vec{v}_i \in R_j) + v_{ij} \cdot \mathbb{I}(\vec{v}_i \notin R_j)\right\} \right) \cdot \mathbb{I}(i^* \geq \ell)\right]$$
$$ \leq \mathbb{E}_{\vec{w} \leftarrow D_j^{n+(m-1)(\ell-1)}} \left[\left(\max \left\{\overline{\varphi}_j(w_{(1)}), w_{(\ell)}\right\}\right) \cdot \mathbb{I}(i^* \geq \ell)\right].$$

Now, consider the case that $i^* < \ell$. In this case, the contribution to the benchmark is exactly $\max\{\overline{\varphi}_j(v_{(1)j}),v_{(i^*)j}\}$. We now wish to argue that, conditioned on $i^*$, $v_{(i^*)j}$ and $v_{(1)j}$, the expected contribution to the right-hand side exceeds this. Indeed, observe that there exists at least one $k$ for which $q'_{i^*,k} > q_{(i^*)}$. Conditioned on $q_{(i^*)}$ (and $v_{(1)j}$, which has no effect), $q'_{i^*,k}$ is simply a uniformly random draw from $[q_{(i^*)},1]$. That is, $w_{(i^*)k'}$, (where $k'=k$ if $k \neq j$ and $k'=m$ if $k = j$) is simply a draw from $D_j$, conditioned on exceeding $w_{(i^*){j}}$. Therefore, $\mathbb{E}[\overline{\varphi}_j(w_{(i^*)k'})| w_{(i^*){j}} = x, w_{(i^*)k'} > w_{(i^*){j}}] \geq x$ (by Fact~\ref{fact:one}). We can now apply Fact~\ref{fact:basic} to conclude that $\mathbb{E}[\max\{\overline{\varphi}_j(v_{(1)j}), v_{(i^*)j}\} \cdot \mathbb{I}(i^* < \ell)] \leq \mathbb{E}[\max\{\overline{\varphi}_j(v_{(1)j}), \max_k \{\overline{\varphi}_j(w_{(i^*)k})\}\} \cdot \mathbb{I}(i^* < \ell)]] \leq \mathbb{E}[\overline{\varphi}_j(w_{(1)})\cdot \mathbb{I}(i^* < \ell)]]$. The last inequality follows simply because ${w_{(1)}} \geq v_{(1)j}$ and also ${w_{(1)} \geq \max_k \{w_{(i^*)k}\}}$.

$$\mathbb{E}_{\vec{v} \leftarrow D^n} \left[\left(\max_{i \in [n]}\left\{\overline{\varphi}_j(v_{ij})^+ \cdot \mathbb{I}(\vec{v}_i \in R_j) + v_{ij} \cdot \mathbb{I}(\vec{v}_i \notin R_j)\right\} \right) \cdot \mathbb{I}(i^* < \ell)\right]$$
$$ \leq \mathbb{E}_{\vec{w} \leftarrow D_j^{n+(m-1)(\ell-1)}} \left[\overline{\varphi}_j(w_{(1)}) \cdot \mathbb{I}(i^* < \ell)\right].$$

Summing up the two left-hand sides yields item $j$'s contribution to the EFFTW benchmark. Summing up the two right-hand sides lower bounds the desired right-hand side.

\end{proof}

\begin{proof}[Proof of Lemma~\ref{lem:w}]
Again couple draws so that $w_i = F_j^{-1}(X_{i,n'})$. Then $\overline{\varphi}_j(X_{(1),n'}) = \overline{\varphi}_j(w_{(1)})$. Additionally, $\mathbb{E}[\overline{\varphi}_j(W_{\ell,n'})|X_{(1),n'} = x, X_{(\ell,n')} = y] \leq F_j^{-1}(y) = w_{(\ell)}$ by Fact~\ref{fact:one}. So the hypotheses of Fact~\ref{fact:basic} are satisfied, and we get that

$$\mathbb{E}_{\vec{w} \leftarrow D^{n'}} \left[\max \left\{\overline{\varphi}_j(w_{(1)}), w_{(\ell)}\right\}\right] 
\leq \mathbb{E}\left[\max\{\overline{\varphi}_j(F_j^{-1}(X_{(1),n'})),\overline{\varphi}_j(F_j^{-1}(W_{\ell,n'}))\}\right] $$ $$=\mathbb{E}\left[\overline{\varphi}_j(F_j^{-1}(\max\{X_{(1),n'},W_{\ell,n'}))\right] =\mathbb{E}\left[\overline{\varphi}_j(F_j^{-1}(X_B(n',\ell)))\right].$$

\end{proof}
\bibliographystyle{alpha}
\bibliography{MasterBib}

\newcommand{\etalchar}[1]{$^{#1}$}
\begin{thebibliography}{EFF{\etalchar{+}}17b}

\bibitem[BCKW15]{BriestCKW15}
Patrick Briest, Shuchi Chawla, Robert Kleinberg, and S.~Matthew Weinberg.
\newblock Pricing lotteries.
\newblock {\em J. Economic Theory}, 156:144--174, 2015.

\bibitem[BGN17]{BabaioffGN17}
Moshe Babaioff, Yannai~A. Gonczarowski, and Noam Nisan.
\newblock The menu-size complexity of revenue approximation.
\newblock In {\em Proceedings of the 49th Annual {ACM} {SIGACT} Symposium on
  Theory of Computing, {STOC} 2017, Montreal, QC, Canada, June 19-23, 2017},
  pages 869--877, 2017.

\bibitem[BILW14]{BabaioffILW14}
Moshe Babaioff, Nicole Immorlica, Brendan Lucier, and S.~Matthew Weinberg.
\newblock A simple and approximately optimal mechanism for an additive buyer.
\newblock In {\em the 55th Annual IEEE Symposium on Foundations of Computer
  Science (FOCS)}, 2014.

\bibitem[BK96]{BulowK96}
Jeremy Bulow and Paul Klemperer.
\newblock Auctions versus negotiations.
\newblock {\em The American Economic Review}, pages 180--194, 1996.

\bibitem[CDW16]{CaiDW16}
Yang Cai, Nikhil Devanur, and S.~Matthew Weinberg.
\newblock A duality based unified approach to bayesian mechanism design.
\newblock In {\em Proceedings of the 48th ACM Conference on Theory of
  Computation(STOC)}, 2016.
\newblock \url{https://arxiv.org/abs/1812.01577}.

\bibitem[CHK07]{ChawlaHK07}
Shuchi Chawla, Jason~D. Hartline, and Robert~D. Kleinberg.
\newblock {Algorithmic Pricing via Virtual Valuations}.
\newblock In {\em the 8th ACM Conference on Electronic Commerce (EC)}, 2007.

\bibitem[CHMS10]{ChawlaHMS10}
Shuchi Chawla, Jason~D. Hartline, David~L. Malec, and Balasubramanian Sivan.
\newblock {Multi-Parameter Mechanism Design and Sequential Posted Pricing}.
\newblock In {\em the 42nd ACM Symposium on Theory of Computing (STOC)}, 2010.

\bibitem[CM16]{ChawlaM16}
Shuchi Chawla and J.~Benjamin Miller.
\newblock Mechanism design for subadditive agents via an ex ante relaxation.
\newblock In {\em Proceedings of the 2016 {ACM} Conference on Economics and
  Computation, {EC} '16, Maastricht, The Netherlands, July 24-28, 2016}, pages
  579--596, 2016.

\bibitem[CMS15]{ChawlaMS15}
Shuchi Chawla, David~L. Malec, and Balasubramanian Sivan.
\newblock The power of randomness in bayesian optimal mechanism design.
\newblock {\em Games and Economic Behavior}, 91:297--317, 2015.

\bibitem[CZ17]{CaiZ17}
Yang Cai and Mingfei Zhao.
\newblock Simple mechanisms for subadditive buyers via duality.
\newblock In {\em Proceedings of the 49th Annual {ACM} {SIGACT} Symposium on
  Theory of Computing, {STOC} 2017, Montreal, QC, Canada, June 19-23, 2017},
  pages 170--183, 2017.

\bibitem[DDT17]{DaskalakisDT17}
Constantinos Daskalakis, Alan Deckelbaum, and Christos Tzamos.
\newblock Strong duality for a multiple-good monopolist.
\newblock {\em Econometrica}, 85(3):735--767, 2017.

\bibitem[EFF{\etalchar{+}}17a]{EdenFFTW17a}
Alon Eden, Michal Feldman, Ophir Friedler, Inbal Talgam{-}Cohen, and S.~Matthew
  Weinberg.
\newblock The competition complexity of auctions: {A} bulow-klemperer result
  for multi-dimensional bidders.
\newblock In {\em Proceedings of the 2017 {ACM} Conference on Economics and
  Computation, {EC} '17, Cambridge, MA, USA, June 26-30, 2017}, page 343, 2017.

\bibitem[EFF{\etalchar{+}}17b]{EdenFFTW17b}
Alon Eden, Michal Feldman, Ophir Friedler, Inbal Talgam{-}Cohen, and S.~Matthew
  Weinberg.
\newblock A simple and approximately optimal mechanism for a buyer with
  complements: Abstract.
\newblock In {\em Proceedings of the 2017 {ACM} Conference on Economics and
  Computation, {EC} '17, Cambridge, MA, USA, June 26-30, 2017}, page 323, 2017.

\bibitem[FFR18]{FeldmanFR18}
Michal Feldman, Ophir Friedler, and Aviad Rubinstein.
\newblock 99{\%} revenue via enhanced competition.
\newblock In {\em Proceedings of the 2018 {ACM} Conference on Economics and
  Computation, Ithaca, NY, USA, June 18-22, 2018}, pages 443--460, 2018.

\bibitem[FLLT18]{FuLLT18}
Hu~Fu, Christopher Liaw, Pinyan Lu, and Zhihao~Gavin Tang.
\newblock The value of information concealment.
\newblock In {\em Proceedings of the Twenty-Ninth Annual {ACM-SIAM} Symposium
  on Discrete Algorithms, {SODA} 2018, New Orleans, LA, USA, January 7-10,
  2018}, pages 2533--2544, 2018.

\bibitem[Har11]{HartlineBook}
Jason~D. Hartline.
\newblock {\em Approximation and Mechanism Design}.
\newblock 2011.

\bibitem[HN13]{HartN13}
Sergiu Hart and Noam Nisan.
\newblock The menu-size complexity of auctions.
\newblock In {\em the 14th ACM Conference on Electronic Commerce (EC)}, 2013.

\bibitem[HN17]{HartN17}
Sergiu Hart and Noam Nisan.
\newblock Approximate revenue maximization with multiple items.
\newblock {\em J. Economic Theory}, 172:313--347, 2017.

\bibitem[HR15]{HartR15}
Sergiu Hart and Philip~J. Reny.
\newblock {Maximizing Revenue with Multiple Goods: Nonmonotonicity and Other
  Observations}.
\newblock {\em Theoretical Economics}, 10(3):893--922, 2015.

\bibitem[LP18]{LiuP18}
Siqi Liu and Christos{-}Alexandros Psomas.
\newblock On the competition complexity of dynamic mechanism design.
\newblock In {\em Proceedings of the Twenty-Ninth Annual {ACM-SIAM} Symposium
  on Discrete Algorithms, {SODA} 2018, New Orleans, LA, USA, January 7-10,
  2018}, pages 2008--2025, 2018.

\bibitem[LY13]{LiY13}
Xinye Li and Andrew Chi-Chih Yao.
\newblock On revenue maximization for selling multiple independently
  distributed items.
\newblock {\em Proceedings of the National Academy of Sciences},
  110(28):11232--11237, 2013.

\bibitem[Mye81]{Myerson81}
Roger~B. Myerson.
\newblock {Optimal Auction Design}.
\newblock {\em Mathematics of Operations Research}, 6(1):58--73, 1981.

\bibitem[Pav11]{Pavlov11}
Gregory Pavlov.
\newblock Optimal mechanism for selling two goods.
\newblock {\em The B.E. Journal of Theoretical Economics}, 11(3), 2011.

\bibitem[RTCY12]{RoughgardenTY12}
Tim Roughgarden, Inbal Talgam-Cohen, and Qiqi Yan.
\newblock Supply-limiting mechanisms.
\newblock In {\em 13th ACM Conference on Electronic Commerce (EC)}, 2012.

\bibitem[RW15]{RubinsteinW15}
Aviad Rubinstein and S.~Matthew Weinberg.
\newblock Simple mechanisms for a subadditive buyer and applications to revenue
  monotonicity.
\newblock In {\em Proceedings of the 16th ACM Conference on Electronic
  Commerce}, 2015.

\bibitem[Tha04]{Thanassoulis04}
John Thanassoulis.
\newblock Haggling over substitutes.
\newblock {\em Journal of Economic Theory}, 117:217--245, 2004.

\bibitem[Yao15]{Yao15}
Andrew Chi-Chih Yao.
\newblock {An n-to-1 bidder reduction for multi-item auctions and its
  applications}.
\newblock In {\em the Twenty-Sixth Annual ACM-SIAM Symposium on Discrete
  Algorithms (SODA)}, 2015.

\end{thebibliography}

\end{document}